\documentclass[11pt, letterpaper]{article}
\pdfoutput=1
\usepackage{fullpage} 
\usepackage{xspace,xcolor,graphicx,tabularx} 
\usepackage{mathtools,amsthm,amssymb} 
\usepackage{enumitem} 
\setenumerate[0]{label={\normalfont (\roman*)}}

\usepackage[T1]{fontenc}
\usepackage[UKenglish]{babel}
\usepackage[scaled=0.86]{helvet}
\usepackage{mathptmx}
\DeclareMathAlphabet{\mathcal}{OMS}{ntxm}{m}{n}
\usepackage{dsfont} 
\let\mathbb\relax
\let\mathbb\mathds
\usepackage{mathrsfs} 
\usepackage{stmaryrd} 
\makeatletter
\renewcommand{\paragraph}{%
  \@startsection{paragraph}{4}%
  {\z@}{2.25ex \@plus 1ex \@minus .2ex}{-1em}%
  {\normalfont\normalsize\bfseries}%
}
\makeatother
\usepackage{authblk}

\definecolor{linkblue}{HTML}{001487}
\usepackage[colorlinks=true,allcolors=linkblue]{hyperref}
\usepackage{url}
\usepackage[nameinlink,capitalize,noabbrev]{cleveref}
\crefname{enumi}{Step}{Steps}

\newtheorem{theorem}{Theorem}[section]
\newtheorem*{theorem*}{Theorem}
\newtheorem{proposition}[theorem]{Proposition}

\newtheorem{lemma}[theorem]{Lemma}

\newtheorem{corollary}[theorem]{Corollary}
\theoremstyle{remark}
\newtheorem{remark}[theorem]{Remark}
\theoremstyle{definition}
\newtheorem{definition}[theorem]{Definition}

\newtheorem{protocol}{Protocol}
\numberwithin{equation}{section}
\newcommand\numberthis{\addtocounter{equation}{1}\tag{\theequation}}

\newcommand{\setft}[1]{\textnormal{#1}}
\newcommand{\eps}{\varepsilon}
\newcommand{\1}{\mathbb{1}}
\newcommand{\id}{\setft{id}}

\newcommand{\N}{\ensuremath{\mathds{N}}}
\newcommand{\R}{\ensuremath{\mathds{R}}}

\newcommand{\bits}{\ensuremath{\{0, 1\}}}

\newcommand{\ot}{\ensuremath{\otimes}}
\newcommand{\deq}{\coloneqq}
\newcommand{\tr}[1]{\mbox{\rm Tr}\!\left[ #1 \right]}
\newcommand{\ptr}[2]{\mbox{\rm Tr}_{#1}\!\left[ #2 \right]}
\newcommand{\pr}[1]{{\rm Pr}\!\left[ #1 \right]}
\newcommand{\prs}[2]{{\rm Pr}_{#1}\!\left[ #2 \right]}
\newcommand{\norm}[1]{\left\lVert#1\right\rVert}
\DeclareMathOperator{\pos}{Pos}

\DeclareMathOperator{\supp}{\setft{supp}}

\newcommand{\sth}{{\setft{~s.t.~}}}
\newcommand{\tand}{\;\textnormal{~and~}\;}

\newcommand{\ket}[1]{|#1\rangle}
\newcommand{\bra}[1]{\langle#1|}
\newcommand{\proj}[1]{\ket{#1}\!\bra{#1}}
\DeclarePairedDelimiterX\braket[2]{\langle}{\rangle}{#1 \delimsize\vert #2}
\newcommand{\cp}{\setft{CP}}
\newcommand{\cptp}{\setft{CPTP}}

\newcommand{\cA}{\ensuremath{\mathcal{A}}}
\newcommand{\cB}{\ensuremath{\mathcal{B}}}
\newcommand{\cC}{\ensuremath{\mathcal{C}}}
\newcommand{\cD}{\ensuremath{\mathcal{D}}}
\newcommand{\cE}{\ensuremath{\mathcal{E}}}
\newcommand{\cF}{\ensuremath{\mathcal{F}}}

\newcommand{\cI}{\ensuremath{\mathcal{I}}}

\newcommand{\cM}{\ensuremath{\mathcal{M}}}
\newcommand{\cN}{\ensuremath{\mathcal{N}}}

\newcommand{\cP}{\ensuremath{\mathcal{P}}}

\newcommand{\cR}{\ensuremath{\mathcal{R}}}
\newcommand{\cS}{\ensuremath{\mathcal{S}}}
\newcommand{\cT}{\ensuremath{\mathcal{T}}}

\newcommand{\cX}{\ensuremath{\mathcal{X}}}
\newcommand{\cY}{\ensuremath{\mathcal{Y}}}
\newcommand{\cZ}{\ensuremath{\mathcal{Z}}}

\DeclareSymbolFont{greekletters}{OML}{ntxmi}{m}{it}
\DeclareMathSymbol{\alpha}{\mathord}{greekletters}{"0B}
\DeclareMathSymbol{\beta}{\mathord}{greekletters}{"0C}
\DeclareMathSymbol{\gamma}{\mathord}{greekletters}{"0D}
\DeclareMathSymbol{\delta}{\mathord}{greekletters}{"0E}
\DeclareMathSymbol{\epsilon}{\mathord}{greekletters}{"0F}
\DeclareMathSymbol{\zeta}{\mathord}{greekletters}{"10}
\DeclareMathSymbol{\eta}{\mathord}{greekletters}{"11}
\DeclareMathSymbol{\theta}{\mathord}{greekletters}{"12}
\DeclareMathSymbol{\iota}{\mathord}{greekletters}{"13}
\DeclareMathSymbol{\kappa}{\mathord}{greekletters}{"14}
\DeclareMathSymbol{\lambda}{\mathord}{greekletters}{"15}
\DeclareMathSymbol{\mu}{\mathord}{greekletters}{"16}
\DeclareMathSymbol{\nu}{\mathord}{greekletters}{"17}
\DeclareMathSymbol{\xi}{\mathord}{greekletters}{"18}
\DeclareMathSymbol{\pi}{\mathord}{greekletters}{"19}
\DeclareMathSymbol{\rho}{\mathord}{greekletters}{"1A}
\DeclareMathSymbol{\sigma}{\mathord}{greekletters}{"1B}
\DeclareMathSymbol{\tau}{\mathord}{greekletters}{"1C}
\DeclareMathSymbol{\upsilon}{\mathord}{greekletters}{"1D}
\DeclareMathSymbol{\phi}{\mathord}{greekletters}{"1E}
\DeclareMathSymbol{\chi}{\mathord}{greekletters}{"1F}
\DeclareMathSymbol{\psi}{\mathord}{greekletters}{"20}
\DeclareMathSymbol{\omega}{\mathord}{greekletters}{"21} 
\DeclareMathSymbol{\vartheta}{\mathord}{greekletters}{"23}
\DeclareMathSymbol{\varpi}{\mathord}{greekletters}{"24}
\DeclareMathSymbol{\varrho}{\mathord}{greekletters}{"25}
\DeclareMathSymbol{\varsigma}{\mathord}{greekletters}{"26}
\DeclareMathSymbol{\varphi}{\mathord}{greekletters}{"27}

\newcommand{\hmin}{H_\setft{min}}
\newcommand{\hmax}{H_{\ensuremath{\text{max}}}}

\newcommand{\dmaxeps}[2]{D_{\setft{max}}^\eps \! \left( #1 \, \middle\Vert \, #2 \right)}
\newcommand{\dalpha}[2]{D_\alpha \! \left( #1 \, \middle\Vert \, #2 \right)}
\newcommand{\dvaralpha}[3]{D_{#1} \! \left( #2 \, \middle\Vert \, #3 \right)}
\newcommand{\dreg}[2]{D^\setft{reg} \! \left( #1 \, \middle\Vert \, #2 \right)}
\newcommand{\drega}[2]{D^\setft{reg}_\alpha \! \left( #1 \, \middle\Vert \, #2 \right)}
\newcommand{\had}{H_\alpha}
\newcommand{\hau}{H^{\shortuparrow}_\alpha}
\newcommand{\hadvar}[1]{H_{#1}}
\newcommand{\hauvar}[1]{H^{\shortuparrow}_{#1}}

\let\H\relax
\newcommand{\H}{\mathcal{H}}
\DeclareMathOperator*{\E}{\mathds{E}}

\usepackage{longfbox}
\newcommand{\mbP}{\ensuremath{\mathds{P}}}

\newcommand*{\freq}[1]{\mathsf{freq}(#1)}
\newcommand*{\Var}[1]{\mathsf{Var}(#1)}

\newcommand*{\Max}[1]{\mathsf{Max}(#1)}
\newcommand*{\MinSigma}[1]{\mathsf{Min}_{\Sigma}(#1)}
\newcommand*{\Min}[1]{\mathsf{Min}(#1)}

\newcommand{\states}{\setft{S}}

\newcommand{\Herm}{\setft{Herm}}

\newcommand{\spec}{\setft{Spec}}
\newcommand{\sym}{\setft{Sym}}

\newcommand{\wexp}{\omega_{\setft{exp}}}

\usepackage{tikz,pgfplotstable}
\usetikzlibrary{fit,calc}
\tikzstyle{porte} = [fill=gray!25, draw]

\begin{document}

\title{Generalised entropy accumulation}

\author[1]{Tony Metger\footnote{Email: \href{mailto:tmetger@ethz.ch}{tmetger@ethz.ch}}}
\author[2]{Omar Fawzi}
\author[3]{David Sutter}
\author[1]{Renato Renner}
\affil[1]{Institute for Theoretical Physics, ETH Zurich, 8093 Zurich, Switzerland}
\affil[2]{Univ Lyon, ENS Lyon, UCBL, CNRS, Inria, LIP, F-69342, Lyon Cedex 07, France}
\affil[3]{IBM Quantum, IBM Research Europe -- Zurich, Switzerland}
\date{\vspace{-1cm}}

\maketitle

\begingroup 
\hyphenpenalty 10000
\exhyphenpenalty 10000
\begin{abstract}
The min-entropy of a quantum system $A$ conditioned on another quantum system $E$ describes how much randomness can be extracted from $A$ with respect to an adversary in possession of $E$.
This quantity plays a crucial role in quantum cryptography: the security proofs of many quantum cryptographic protocols reduce to showing a lower bound on such a min-entropy.
Here, we develop a new tool, called \emph{generalised entropy accumulation}, for computing such bounds.
Concretely, we consider a sequential process in which each step outputs a system $A_i$ and updates a side information register $E$. We prove that if this process satisfies a natural ``non-signalling'' condition between past outputs and future side information, the min-entropy of the outputs $A_1, \dots, A_n$ conditioned on the  side information $E$ at the end of the process can be bounded from below by a sum of von Neumann entropies associated with the individual steps. This is a generalisation of the entropy accumulation theorem (EAT)~\cite{eat}, which deals with a more restrictive model of side information: there, past side information cannot be updated in subsequent rounds, and newly generated side information has to satisfy a Markov condition. 

Due to its more general model of side-information, our generalised EAT can be applied more easily and to a broader range of cryptographic protocols. 
In particular, it is the first general tool that is applicable to mistrustful device-independent cryptography.
To demonstrate this, we give the first security proof for blind randomness expansion~\cite{miller2017randomness} against general adversaries.
Furthermore, our generalised EAT can be used to give improved security proofs for quantum key distribution~\cite{eat_for_qkd}, and also has applications beyond quantum cryptography.
\end{abstract}
\endgroup

\newpage 
\thispagestyle{empty}
{
\hypersetup{linkcolor=black}
\tableofcontents
}

\newpage

\section{Introduction} \label{sec:intro}

Suppose that Alice and Eve share a quantum state $\rho_{A^n E}$.
From her systems $A^n \deq A_1 \dots A_n$, Alice would like to extract bits that look uniformly random to Eve, except with some small failure probability $\eps$.
The number of such random bits that Alice can extract is given by the smooth min-entropy $\hmin^\eps(A^n|E)_\rho$~\cite{renner_thesis}.
This quantity plays a central role in quantum cryptography: for example, the main task in security proofs of quantum key distribution (QKD) protocols is usually finding a lower bound for the smooth min-entropy.

Unfortunately, for many cryptographic protocols deriving such a bound is challenging.
Intuitively, the reason is the following: the state $\rho_{A^n E}$ is usually created as the output of a multi-round protocol, where each round produces one of Alice's systems $A_i$ and allows Eve to execute some attack to gain information about $A_1, \dots, A_i$.
These attacks can depend on each other, i.e., Eve may use what she learnt in round $i - 1$ to plan her attack in round $i$.
This non-i.i.d.~nature of the attacks makes it hard to find a lower bound on $\hmin^\eps(A^n|E)_\rho$ that holds for any possible attack that Eve can execute.
In contrast, it is typically much easier to compute a conditional von Neumann entropy associated with a single-round of the protocol, where the non-i.i.d.~nature of Eve's attack plays no role.
Therefore, it is desirable to relate the smooth min-entropy of the output of the multi-round protocol to the von Neumann entropies associated with the individual rounds.

From an information-theoretic point of view, this question can be phrased as follows:
can the smooth min-entropy $\hmin^\eps(A^n | E)_\rho$ be bounded from below in terms of von Neumann entropies \smash{$H(A_i | E_i)_{\rho^i_{A_i E_i}}$} for some (yet to be determined) systems $E_i$ and states $\rho^i_{A_i E_i}$ related to $\rho$?
While for general states $\rho_{A^n E}$ no useful lower bound can be found, previous works have established such bounds under additional assumptions on the state $\rho_{A^n E}$.

The first bound of this form was proven via the \emph{asymptotic equipartition property} (AEP)~\cite{tomamichel09}. It assumes that the system $E$ is $n$-partite (i.e., we replace $E$ by $E^n = E_1 \dots E_n$) and that the state $\rho_{A^n E^n} = \rho_{A_1 E_1} \ot \dots \ot \rho_{A_n E_n}$ is a product of identical states.
Then, the AEP shows that\footnote{Since $\rho$ is a product of identical states, all of the terms $H(A_i | E_i)_\rho$ are equal, i.e., $\sum_{i=1}^n H(A_i | E_i)_\rho = n H(A_i | E_i)_\rho$ for any $i$. We write the sum here explicitly to highlight the analogy with the EAT presented below.}
\begin{align*}
H_{\min}^{\eps}(A^n  | E^n)_\rho \geq \sum_{i=1}^n H(A_i | E_i)_\rho - O(\sqrt{n}) \ .
\end{align*}
For applications in cryptography, the assumption that $\rho$ is an i.i.d.~product state is usually too strong: it corresponds to the (unrealistic) assumption that Eve executes the same independent attack in each round, a so-called \emph{collective attack}.

The \emph{entropy accumulation theorem} (EAT)~\cite{eat} is a generalisation of the AEP which requires far weaker assumptions on the state $\rho_{A^nE}$.
Specifically, the EAT considers states that result from a sequential process that starts with a state $\rho^0_{R_0 E'}$ and in every step outputs a system $A_i$ and a piece of side information $I_i$. 
The system $E'$ is not acted upon during the process.
The full side information at the end of this process is $E = I_1 \dots I_n E'$.
We can represent such a process by the following diagram, where $\cM_i$ are quantum channels.
 
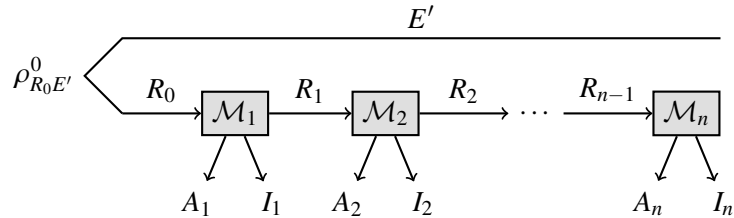
\begin{center}
    \begin{tikzpicture}[thick]
        \draw
            (0, 0) node[porte] (m1) {$\cM_1$}
            ++(2, 0) node[porte] (m2) {$\cM_2$}
            ++(2, 0) node (dotdotdot) {$\cdots$}
            ++(2, 0) node[porte] (mn) {$\cM_n$}
            (m1) ++(-.5, -1.2) node (a1) {$A_1$}
            (m1) ++(.5, -1.2) node (b1) {$I_1$}
            (m2) ++(-.5, -1.2) node (a2) {$A_2$}
            (m2) ++(.5, -1.2) node (b2) {$I_2$}
            (mn) ++(-.5, -1.2) node (an) {$A_n$}
            (mn) ++(.5, -1.2) node (bn) {$I_n$}
            ;
       \draw 
            (m1) ++(-1.5, 0) edge[->] node[midway, above] {$R_0$} (m1)
            (m1) edge[->] node[midway, above] {$R_1$} (m2)
            (m2) edge[->] node[midway, above] {$R_2$} (dotdotdot)
            (dotdotdot) edge[->] node[midway, above] {$R_{n-1}$} (mn)
            (m1) edge[->] (a1)
            (m1) edge[->] (b1)
            (m2) edge[->] (a2)
            (m2) edge[->] (b2)
            (mn) edge[->] (an)
            (mn) edge[->] (bn)
            ;
        \draw
        (m1) ++(-1.5, 0) to  ++(-.5, .5) node[left] {$\rho^0_{R_0 E'}$} to ++(.5, .5) coordinate (topright) to node[midway, above] {${E'}$} (topright -| mn.east) coordinate (rightwall)
            ;            
    \end{tikzpicture}
    \end{center}

The EAT requires an additional condition on the side information: the new side information $I_i$ generated in round $i$ must be independent from the past outputs $A^{i-1}$ conditioned on the existing side information $I^{i-1} E'$.
Mathematically, this is captured by the condition that the systems $A^{i-1} \leftrightarrow  I^{i-1} E' \leftrightarrow I_i$ form a Markov chain for any initial state $\rho^0_{R_0 E'}$.
With this Markov condition, the EAT states that\footnote{The EAT from~\cite{eat} also makes an analogous statement about an upper bound on the max-entropy $\hmax$. 
We derive a generalisation of that statement in~\cref{app:hmax} but only focus on $\hmin$ in the introduction and main text since that is the case that is typically relevant for applications.}
\begin{align}
H_{\min}^{\eps}(A^n | I^n E')_{\cM_n \circ \dots \circ \cM_1(\rho^0_{R_0 E' })} \geq \sum_{i=1}^n \inf_{\omega} H(A_i | I_i \tilde E)_{\cM_i(\omega)} - O(\sqrt{n}) \ , \label{eqn:intro_eat}
\end{align}
where $\tilde E$ is a purifying system isomorphic to $R_{i-1}$ and the infimum is taken over all states $\omega$ on systems $R_{i-1}\tilde E$.\footnote{In fact, the EAT is more general in that it allows taking into account observed statistics to restrict the minimization over $\omega_{A_i B_i E}$, but we restrict ourselves to the simpler case without statistics in this introduction.}

Let us discuss the model of side information used by the EAT in more detail.
The EAT considers side information consisting of two parts: the initial side information $E'$ (which is not acted upon during the process) and the outputs $I^n = I_1 \dots I_n$.
This splitting of side information into a ``static'' part $E'$ and a part $I^n$ which is generated in each step of the process is particularly suited to device-independent cryptography: there, Eve prepares a device in an initial state $\rho^0_{R_0 E'}$, where $R_0$ is the device's internal memory and $E'$ is Eve's initial side information from preparing the device.
Then, Alice (and Bob, though we only consider Alice's system here) executes a multi-round protocol with this device, where each round leaks some additional piece of information $I_i$ to Eve, so that Eve's side information at the end of the protocol is $I^n E'$.
Indeed, the EAT has been used to establish tight security proofs in the device-independent setting, see e.g.,~\cite{arnon2018practical,BMP18}.

The Markov condition in the EAT captures the following intuition:
if we want to find a bound on $H_{\min}^{\eps}(A^n | I^n E')$ in terms of single-round quantities, it is required that side information about $A_i$ is itself output in step $i$, as otherwise we cannot hope to estimate the contribution to the total entropy from step $i$.
To illustrate what could happen without such a condition, consider a case where $A_i$ is classical and no side information is output in the first $n-1$ rounds, but the side information $I_n$ in the last round contains a copy of the systems $A^n$ (which can be passed along during the process in the systems $R_i$).
Then, clearly $H_{\min}^{\eps}(A^n | I^n E') = 0$, but for the first $n-1$ rounds, each single-round entropy bound that only considers the systems $A_i$ and $I_i$ can be positive.

\paragraph{Main result.} 
In this work, we further relax the assumptions on how the final state $\rho_{A^n E}$ is generated.
Specifically, we consider sequential processes as in the EAT, but with a fully general model of side information, i.e., the side information can be updated in each step in the process.
Diagrammatically, such a process can be represented as follows:
    
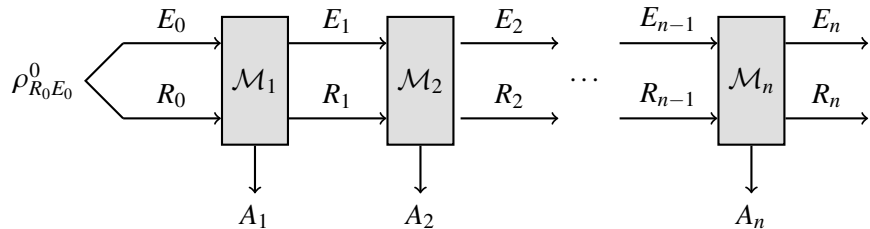
\begin{center}
    \begin{tikzpicture}[thick]
        \draw
            (0, 0) node[porte, minimum height=1.7cm] (m1) {$\cM_1$}
            ++(2.2, 0) node[porte, minimum height=1.7cm] (m2) {$\cM_2$}
            ++(2.2, 0) node (dotdotdot) {$\cdots$}
            ++(2.2, 0) node[porte, minimum height=1.7cm] (mn) {$\cM_n$}
            (m1) ++(0, -1.8) node (a1) {$A_1$}
            (m2) ++(0, -1.8) node (a2) {$A_2$}
            (mn) ++(0, -1.8) node (an) {$A_n$}
            ;
        \draw
            (m1.west) ++(0, 0.5) coordinate (m1-in1)
            (m1.west) ++(0, -0.5) coordinate (m1-in2)
            (m2.west) ++(0, 0.5) coordinate (m2-in1)
            (m2.west) ++(0, -0.5) coordinate (m2-in2)
            (mn.west) ++(0, 0.5) coordinate (mn-in1)
            (mn.west) ++(0, -0.5) coordinate (mn-in2)
            (mn.east) ++(0, 0.5) coordinate (mn-out1)
            (mn.east) ++(0, -0.5) coordinate (mn-out2)
            (dotdotdot.west) ++(0, 0.5) coordinate (dotdotdot-in1)
            (dotdotdot.west) ++(0, -0.5) coordinate (dotdotdot-in2)
            ;
       \draw 
            (m1-in1) ++(-1.3, 0) edge[->] node[midway, above] {$E_0$} (m1-in1)
            (m1-in2) ++(-1.3, 0) edge[->] node[midway, above] {$R_0$} (m1-in2)
            (m2-in1) ++(-1.3, 0) edge[->] node[midway, above] {$E_1$} (m2-in1)
            (m2-in2) ++(-1.3, 0) edge[->] node[midway, above] {$R_1$} (m2-in2)
            (dotdotdot-in1) ++(-1.3, 0) edge[->] node[midway, above] {$E_2$} (dotdotdot-in1)
            (dotdotdot-in2) ++(-1.3, 0) edge[->] node[midway, above] {$R_2$} (dotdotdot-in2)
            (mn-in1) ++(-1.3, 0) edge[->] node[midway, above] {$E_{n-1}$} (mn-in1)
            (mn-in2) ++(-1.3, 0) edge[->] node[midway, above] {$R_{n-1}$} (mn-in2)
            (mn-out1) edge[->] node[midway, above] {$E_{n}$} coordinate (mn-out1) ++(1.1, 0)
            (mn-out2) edge[->] node[midway, above] {$R_{n}$} coordinate (mn-out2) ++(1.1, 0)
            (m1) edge[->] (a1)
            (m2) edge[->] (a2)
            (mn) edge[->] (an)
            ;
        \draw
        (m1-in2) ++(-1.3, 0) to  ++(-.5, .5) node[left] {$\rho^0_{R_0 E_0}$} to ++(.5, .5)
            ;            
    \end{tikzpicture}
    \end{center}

Our generalised EAT then states the following. 
\begin{theorem} \label{thm:intro_eat}
Consider quantum channels $\cM_i: R_{i-1} E_{i-1} \to A_i R_i E_i$ that satisfy the following ``non-signalling'' condition (discussed in detail below): for each $\cM_i$, there must exist a quantum channel $\cR_i: E_{i-1} \to E_i$ such that 
\begin{align}
\setft{Tr}_{A_i R_i} \circ \cM_i = \cR_i \circ \setft{Tr}_{R_{i-1}} \,. \label{eqn:intro_ns}
\end{align}
Then, the min-entropy of the outputs $A^n$ conditioned on the final side information $E_n$ can be bounded as 
\begin{align}
\hmin^\eps(A^n | E_n )_{\cM_n \circ \dots \circ \cM_1(\rho^0_{R_0 E_0 })} \geq \sum_{i = 1}^n \inf_{\omega} H(A_i|E_i \tilde E_{i-1} )_{\cM_i(\omega)} - O(\sqrt{n}) \, , \label{eqn:intro_gen_eat}
\end{align}
where $\tilde E_{i-1} \equiv R_{i-1}E_{i-1}$ is a purifying system for the input to $\cM_i$ and the infimum is taken over all states $\omega$ on systems $R_{i-1}E_{i-1}\tilde E_{i-1}$.\footnote{As usual, the channels $\cM_i$ act as identity on any additional systems that may be part of the input state, i.e.~$\cM_i(\omega_{R_{i-1}E_{i-1}\tilde E_{i-1}}) = (\cM_i \ot \id_{\tilde E_{i-1}})(\omega_{R_{i-1}E_{i-1}\tilde E_{i-1}})$ is a state on $A_i R_i E_i \tilde E_{i-1}$. In particular, the register $\tilde E_{i-1}$ containing a purification of the input is also part of the output state.}
\end{theorem}
We give a formal statement and proof in \cref{sec:eat} and also show that, similarly to the EAT, statistics collected during the process can be used to restrict the minimization over $\omega$ (see \cref{thm:with_testing} for the formal statement).
By a simple duality argument, \cref{eqn:intro_gen_eat} also implies an upper bound on the smooth max-entropy $\hmax$, which we explain in \cref{app:hmax}.
This generalises a similar result from~\cite{eat}, although in~\cite{eat} one could not make use of duality due to the Markov condition and instead had to prove the statement about $\hmax$ separately, again highlighting that our generalised EAT is easier to work with.

The intuition behind the non-signalling condition in our generalised EAT is similar to the Markov condition in the original EAT:
by the same reasoning as for the Markov condition, since the lower bound is made  up of terms of the form $H(A_i|E_i \tilde E_{i-1} )_{\cM_i(\omega)}$, it is required that side information about $A_i$ that is present in the final system $E_n$ is already present in $E_i$.
This means that side information about $A_i$ should not be passed on via the $R$-systems and later be included in the $E$-systems.
The non-signalling condition captures this requirement: it demands that if one only considers the marginal of the new side information $E_i$ (without the new output $A_i$), it must be possible to generate this state from the past side information $E_{i-1}$ alone, without access to the system $R_{i-1}$.
This means that any side information that $E_i$ contains about the past outputs $A_1 \dots A_{i-1}$ must have essentially already been present in $E_{i-1}$ and could not have been stored in $R_{i-1}$.

The name ``non-signalling condition'' is due to the fact that \cref{eqn:intro_ns} is a natural generalisation of the standard non-signalling conditions in non-local games:
if we view the systems $R_{i-1}$ and $R_i A_i$ as the inputs and outputs on ``Alice's side'' of $\cM_i$, and $E_{i-1}$ and $E_i$ as the inputs and outputs on ``Eve's side'', then \cref{eqn:intro_ns} states that the marginal of the output on Eve's side cannot depend on the input on Alice's side.
This is exactly the non-signalling condition in non-local games, except that here the inputs and outputs are allowed to be fully quantum.

To understand the relation between the Markov and non-signalling conditions, it is instructive to consider the setting of the original EAT as a special case of our generalised EAT.
In the original EAT, the full side information available after step $i$ is $E' I^{i}$, and past side information is not updated during the process.
For our generalised EAT, we therefore set $E_i = E' I^i$ and consider maps $\cM_i = \cM'_i \ot \id_{E_{i-1}}$, where $\cM'_i: R_{i-1} \to A_i I_i R_i$ is the map used in the original EAT.
We need to check that with this choice of systems and maps, the Markov condition of the original EAT implies the non-signalling condition of our generalised EAT.
The Markov condition requires that for any state input \smash{$\omega^{i-1}_{A^{i-1}I^{i-1}R_{i-1}E'}$}, the output state \smash{$\omega^{i}_{A^i I^i R_i E'} = \cM_i(\omega^{i-1})$} satisfies $A^{i-1} \leftrightarrow  I^{i-1} E' \leftrightarrow I_i$.\footnote{Strictly speaking, the EAT as stated in~\cite{eat} only requires that this Markov property holds for any input state $\omega^{i-1}$ in the image of the previous maps $\cM_{i-1} \circ \dots\circ \cM_1$. The same is true for the non-signalling condition, i.e., one can check that our proof of the generalised EAT still works if the map $\cR_i$ only satisfies \cref{eqn:intro_ns} on states in the image of $\cM_{i-1} \circ \dots\circ \cM_1$. To simplify the presentation, we use the stronger condition \cref{eqn:intro_ns} throughout this paper.}
It is then a standard result on quantum Markov chains~\cite{Pet86} that there must exist a quantum channel $\cR_i: I^{i-1} E' \to I^{i} E'$ such that \smash{$\omega^i_{I^i E'} = \cR_i(\omega^{i-1}_{I^{i-1} E'})$}.
Remembering that we defined $E_i = E' I^i$ (so that $\cR_i: E_{i-1} \to E_i$) and adding the systems $A^{i-1}$ (on which both $\cM_i$ and $\cR_i$ act as identity), we find that $\cM_i$ satisfies the non-signalling condition: 
\begin{align*}
\setft{Tr}_{A_i R_i} \circ \cM_i(\omega^{i-1}_{A^{i-1} R_{i-1} E_{i-1}}) = \omega^i_{A^{i-1} E_i} = \cR_i(\omega^{i-1}_{A^{i-1} E_{i-1}}) = \cR_i \circ \setft{Tr}_{R_{i-1}} (\omega^{i-1}_{A^{i-1} R_{i-1} E_{i-1}}) \,.
\end{align*}
Then, noting that all conditioning systems on which $\cM_i$ acts as the identity map can collectively be replaced by a single purifying system isomorphic to the input, we see that we recover the original EAT (\cref{eqn:intro_eat}) from our generalised EAT (\cref{eqn:intro_gen_eat}).

We emphasise that while the original EAT with the Markov condition can be recovered as a special case, our model of side information and the non-signalling condition are much more general than the original EAT; arguably, for a sequential process they are the most natural and general way of expressing the notion that future side information should not contain new information about past outputs, which appears to be necessary for an EAT-like result.
To demonstrate the greater generality of our result, in \cref{sec:application_eat} we use it to give the first multi-round proof for blind randomness expansion, a task to which the original EAT could not be applied, and a more direct proof of the E91 QKD protocol than was possible with the original EAT.
Our generalised EAT can also be used to prove security of a much larger class of QKD protocols than the original EAT. 
Interestingly, for (device-dependent) QKD protocols, no ``hidden system'' $R$ is needed and therefore the non-signalling condition is trivially satisfied, i.e., the advantage of our generalised EAT for QKD security proofs stems entirely from the more general model of side information, not from replacing the Markov condition by the non-signalling condition; see~\cref{sec:qkd} for an informal comparison of how the original and generalised EAT can be applied to QKD, and~\cite{eat_for_qkd} for a detailed treatment of the application of our generalised EAT to QKD, including protocols to which the original EAT could not be applied.

\paragraph{Proof sketch.}
The generalised EAT involves both the min-entropy, which can be viewed as a ``worst-case entropy'', and the von Neumann entropy, which can be viewed as an ``average case entropy''.
These two entropies are special cases of a more general family of entropies called R\'enyi entropies, which are denoted by $H_\alpha$ for a parameter $\alpha > 1$ (see \cref{sec:entropies_prelim} for a formal definition).\footnote{We note that the definition of R\'enyi entropies can be extended to $\alpha < 1$, but we will only need the case $\alpha > 1$.}
The min-entropy can be obtained from the R\'enyi entropy by taking $\alpha \to \infty$, whereas the von Neumann entropy corresponds to the limit $\alpha \to 1$.
Hence, the R\'enyi entropies interpolate between the min-entropy and the von Neumann entropy, and they will play a crucial role in our proof.

The key technical ingredient for our generalised EAT is a new chain rule for R\'enyi entropies (\cref{thm:improved_entropy_chain_rule} in the main text). 
\begin{lemma} \label{lem:chain_rule_intro}
Let $\alpha \in (1, 2)$, $\rho_{ARE}$ a quantum state, and $\cM: RE \to A'R'E'$ a quantum channel which satisfies the non-signalling condition in \cref{eqn:intro_ns}, i.e.~there exists a channel $\cR: E \to E'$ such that $\setft{Tr}_{A'R'} \circ \cM = \cR \circ \setft{Tr}_{R}$. 
Then 
\begin{align}
\had(AA'|E')_{\cM(\rho)} \geq \had(A|E)_{\rho} + \inf_{\omega_{RE\tilde E}}\hadvar{\frac{1}{2-\alpha}}(A'|E' \tilde E)_{\cM(\omega)}
\end{align}
for a purifying system $\tilde E \equiv RE$, where the infinimum is over all quantum states $\omega$ on systems $RE\tilde E$.
\end{lemma}
We first describe how this chain rule implies our generalised EAT, following the same idea as in \cite{eat, dupuis2019entropy}.
For this, recall that our goal is to find a lower bound on $\hmin^\eps(A^n | E_n )_{\cM_n \circ \dots \circ \cM_1(\rho^0_{R_0 E_0 })}$ for a sequence of maps satisfying the non-signalling condition $\setft{Tr}_{A_i R_i} \circ \cM_i = \cR_i \circ \setft{Tr}_{R_{i-1}}$.
As a first step, we use a known relation between the smooth min-entropy and the R\'enyi entropy~\cite{tomamichel09}, which (up to a small penalty term depending on $\eps$ and $\alpha$) reduces the problem to lower-bounding 
\begin{align*}
H_\alpha(A^n | E_n )_{\cM_n \circ \dots \circ \cM_1(\rho^0_{R_0 E_0})} = H_\alpha(A_n A^{n-1} | E_n )_{\cM_n \circ \dots \circ \cM_1(\rho^0_{R_0 E_0})} \,.
\end{align*}
To this, we can apply \cref{lem:chain_rule_intro} by choosing $A = A^{n-1}$, $A' = A_n$, $E = E_{n-1}$, $E' = E_n$, $R = R_{n-1}$, $R' = R_n$, and $\rho = \cM_{n-1} \circ \dots \circ \cM_1(\rho^0_{R_0 E_0})$.
Then, since the map $\cM_n$ satisfies the non-signalling condition, \cref{lem:chain_rule_intro} implies that 
\begin{align*}
&\had(A_1^n | E_n )_{\cM_n \circ \dots \circ \cM_1(\rho_{R_0 E_0 })} \\
&\hspace{20mm}\geq \had(A_1^{n-1} | E_{n-1} )_{\cM_{n-1} \circ \dots \circ \cM_1(\rho_{R_0 E_0 })}  + \inf_{\omega \in \states(R_{n-1}E_{n-1}\tilde E_{n-1})} \hadvar{\frac{1}{2-\alpha}}(A_n|E_n \tilde E_{n-1})_{\cM_n(\omega)} \,.
\end{align*}
We can now repeat this argument for the term $\had(A_1^{n-1} | E_{n-1} )_{\cM_{n-1} \circ \dots \circ \cM_1(\rho_{R_0 E_0 })}$. 
After $n$ applications of \cref{lem:chain_rule_intro}, we find that 
\begin{align*}
\had(A_1^n | E_n )_{\cM_n \circ \dots \circ \cM_1(\rho_{R_0 E_0 })} 
&\geq \sum_{i = 1}^n \inf_{\omega \in \states(R_{i-1}E_{i-1}\tilde E_{i-1})} \hadvar{\frac{1}{2-\alpha}}(A_i|E_i \tilde E_{i-1})_{\cM_i(\omega)} \,.
\end{align*}
To conclude, we use a continuity bound from \cite{dupuis2019entropy} to relate $\hadvar{\frac{1}{2-\alpha}}(A_i|E_i \tilde E_{i-1})_{\cM_i(\omega)}$ to $H(A_i|E_i \tilde E_{i-1})_{\cM_i(\omega)}$.
It can be shown that for a suitable choice of $\alpha$, the penalty terms we incur by switching from the min-entropy to the R\'enyi entropy and then to the von Neumann entropy scale as $O(\sqrt{n})$.
Therefore, we obtain \cref{eqn:intro_gen_eat}.
We also provide a version that allows for ``testing'' (which is crucial for application in quantum cryptography and explained in detail in \cref{sec:testing}) and features explicit second-order terms similar to those in~\cite{dupuis2019entropy}.

We now turn our attention to the proof of \cref{lem:chain_rule_intro}.
For this, we need to introduce the (sandwiched) R\'enyi divergence of order $\alpha$ between two (possibly unnormalised) quantum states $\rho$ and $\sigma$, denoted by $\dalpha{\rho}{\sigma}$.
We refer to \cref{sec:entropies_prelim} for a formal definition; for this overview, it suffices to know that $\dalpha{\rho}{\sigma}$ is a measure of how different $\rho$ is from $\sigma$, and that the conditional R\'enyi entropy is related to the R\'enyi divergence by 
\begin{align*}
H_{\alpha}(A|B)_{\rho} = - \dalpha{\rho_{AB}}{\1_A \ot \rho_B} \,.
\end{align*}
Our starting point for proving \cref{lem:chain_rule_intro} is the following chain rule for the R\'enyi divergence from~\cite{sharp_divergence}:
\begin{align}
\dalpha{\cM(\rho)}{\cF(\sigma)} \leq \dalpha{\rho_{ARE}}{\sigma_{ARE}} + \lim_{n \to \infty} \frac{1}{n} \sup_{\omega_{R^n E^n \tilde E^n}}\dalpha{\cM^{\ot n}(\omega)}{\cF^{\ot n}(\omega)} \,, \label{eqn:intro_old_chain_rule}
\end{align}
where $\cM$ and $\cF$ are (not necessarily trace preserving) quantum channels from $RE$ to $A'R'E'$, and $\rho$ and $\sigma$ are any quantum states on $ARE$.
The optimization is over all quantum states $\omega$ on $n$ copies of the systems $RE\tilde E$ (with $\tilde E \equiv RE$ as before).

Making a suitable choice of $\cF$ (which depends on $\cM$) and $\sigma$ (which depends on $\rho$), one can turn \cref{eqn:intro_old_chain_rule} into the following chain rule for the conditional R\'enyi entropy: 
\begin{align}
\had(AA'|E')_{\cM(\rho)} \geq \had(A|RE)_{\rho} + \lim_{n \to \infty} \frac{1}{n} \inf_{\omega_{R^n E^n \tilde E^n}} \had((A')^n|(E')^n \tilde E^n)_{\cM^{\ot n}(\omega)} \,. \label{eqn:intro_h_old}
\end{align}
This chain rule resembles \cref{lem:chain_rule_intro}, but is significantly weaker and cannot be used to prove a useful entropy accumulation theorem.
The reason for this is twofold: 
\begin{enumerate}
\item \cref{eqn:intro_h_old} provides a lower bound in terms of $\had(A|RE)$, not $\had(A|E)$.
The additional conditioning on the $R$-system can drastically lower the entropy: for example, in a device-independent scenario, $R$ would describe the internal memory of the device. 
Then, Alice's output $A$ contains no entropy when conditioned on the internal memory of the device that produced the output, i.e.~$\had(A|RE) = 0$.
On the other hand, Alice's output conditioned only on Eve's side information $E$ may be quite large (and can usually be certified by playing a non-local game), i.e.~$\had(A|E) > 0$.
\item \cref{eqn:intro_h_old} contains the regularised quantity $\lim_{n \to \infty} \frac{1}{n} \inf_{\omega_{R^n E^n \tilde E^n}} \had((A')^n|(E')^n \tilde E^n)_{\cM^{\ot n}(\omega)}$. Due to the limit $n \to \infty$, this quantity cannot be computed numerically and therefore the bound in \cref{eqn:intro_h_old} cannot be evaluated for concrete examples.
\end{enumerate}
We now describe how we overcome each of these issues in turn.
\begin{enumerate}
\item We prove a new variant of Uhlmann's theorem~\cite{Uhlmann76}, a foundational result in quantum information theory.
The original version of Uhlmann's theorem deals with the case of $\alpha = 1/2$; we show that for $\alpha > 1$, a similar result holds, but an additional regularisation is required.
Concretely, we prove that for any states $\rho_{ARE}$ and $\sigma_{AE}$:  
\begin{align}
\lim_{k \to \infty} \frac{1}{k} \inf_{\substack{\hat \sigma_{A^k R^k E^k} \\\sth \hat \sigma_{A^k E^k} = \sigma_{AE}^{\otimes k}}} \dalpha{\rho_{ARE}^{\otimes k}}{\hat \sigma_{A^k R^k E^k}} = \dalpha{\rho_{AE}}{\sigma_{AE}} \,. \label{eqn:intro_uhlmann}
\end{align}
The proof of this result relies heavily on the spectral pinching technique~\cite{tomamichel2015quantum, Sutter_book} and we refer to \cref{lem:min_over_ref_system} for details as well as a non-asymptotic statement with explicit error bounds.

We make use of this extended Uhlmann's theorem as follows: for the case we are interested in, the map $\cF$ in \cref{eqn:intro_old_chain_rule} satisfies a non-signalling condition.
We can show that this condition implies that for any state $\hat \sigma_{A^k R^k E^k} \sth \hat \sigma_{A^k E^k} = \sigma_{AE}^{\otimes k}$: 
\begin{align*}
\dalpha{\cM(\rho)}{\cF(\sigma)} = \frac{1}{k} \dalpha{\cM^{\ot k}(\rho_{ARE}^{\otimes k})}{\cF^{\ot k}(\hat \sigma_{A^k R^k E^k})} \,.
\end{align*}
Applying \cref{eqn:intro_old_chain_rule} to the r.h.s.~of this equality results in a bound that contains $\dalpha{\rho_{ARE}^{\otimes k}}{\hat \sigma_{A^k R^k E^k}}$.
We can now minimise over all states $\hat \sigma_{A^k R^k E^k} \sth \hat \sigma_{A^k E^k} = \sigma_{AE}^{\otimes k}$ and take the limit $k \to \infty$.
Then, \cref{eqn:intro_uhlmann} allows us to drop the $R$-system.
Therefore, under the non-signalling condition on $\cF$, we obtain the following improved chain rule for the sandwiched R\`enyi divergence, which might be of independent interest:
\begin{align*}
\dalpha{\cM(\rho)}{\cF(\sigma)} \leq \dalpha{\rho_{AE}}{\sigma_{AE}} + \lim_{n \to \infty} \frac{1}{n} \sup_{\omega_{R^n E^n \tilde E^n}}\dalpha{\cM^{\ot n}(\omega)}{\cF^{\ot n}(\omega)} \,.
\end{align*}
Using this chain rule, we can show that \cref{eqn:intro_h_old} still holds if $\had(A|RE)$ is replaced by $\had(A|E)$.
\item To remove the need for a regularisation in \cref{eqn:intro_h_old}, we show that due to the permutation-invariance of $\cM^{\ot n}$ and $\cF^{\ot n}$, for $\alpha > 1$ and $n \to \infty$ one can replace the optimization over $\omega_{R^n E^n \tilde E^n}$ with a fixed input state, namely the projector onto the symmetric subspace of $R^n E^n \tilde E^n$.
For this replacement, one incurs a small loss in $\alpha$, replacing it by $\frac{1}{2 - \alpha}$ (which is only slightly larger than $\alpha$ in the typical regime where $\alpha$ is close to 1).
The projector onto the symmetric subspace has a known representation as a mixture of tensor product states~\cite{christandl2009postselection}.
Combining these two steps, we show that the optimization over $\omega_{R^n E^n \tilde E^n}$ can be restricted to tensor product states, which means that the regularisation in \cref{eqn:intro_h_old} can be removed (see \cref{sec:removing_reg} for details): 
\begin{align*}
\lim_{n \to \infty} \frac{1}{n} \inf_{\omega_{R^n E^n \tilde E^n}} \had((A')^n|(E')^n \tilde E^n)_{\cM^{\ot n}(\omega)} \geq \inf_{\omega_{RE\tilde E}} \hadvar{\frac{1}{2-\alpha}}(A'|E' \tilde E)_{\cM(\omega)} \,.
\end{align*}
\end{enumerate}
Combining these results yields \cref{lem:chain_rule_intro} and, as a result, our generalised EAT.

\paragraph{Sample application: blind randomness expansion.}
The main advantage of the generalised EAT over previous  results is its broader applicability.
For example, as demonstrated in~\cite{eat_for_qkd}, the generalised EAT can be used to prove the security of prepare-and-measure QKD protocols, which is of immediate practical relevance, and can also simplify the analysis of entanglement-based QKD protocols as discussed in \cref{sec:qkd}.
Here, we focus on the application of our generalised EAT to mistrustful device-independent (DI) cryptography.
In mistrustful DI cryptography, multiple parties each use a quantum device to execute a protocol with one another.
Each party trusts neither its quantum device nor the other parties in the protocol.
Hence, from the point of view of one party, say Alice, all the remaining parties in the protocol are collectively treated as an adversary Eve, who may also have prepared Alice's untrusted device.

While the original EAT could be used to analyse DI protocols in which the parties trust each other, e.g.~DIQKD~\cite{arnon2019simple}, the setting of mistrustful DI cryptography is significantly harder to analyse because the adversary Eve actively participates in the protocol and may update her side information during the protocol in arbitrary ways.
Analysing such protocols requires the more general model of side information we deal with in this paper.
As a concrete example for mistrustful DI cryptography, we consider blind randomness expansion, a primitive introduced in~\cite{miller2017randomness}.
Previous work~\cite{miller2017randomness,fu2018local} could only analyse blind randomness expansion under the i.i.d.~assumption.
Here, we give the first proof that blind randomness expansion is possible for general adversaries.
The proof is a straightforward application of our generalised EAT and briefly sketched below; we refer to \cref{sec:blind_re} for a detailed treatment.

In blind randomness expansion, Alice receives an untrusted quantum device from the adversary Eve.
Alice then plays a non-local game, e.g.~the CHSH game, with this device and Eve, and wants to extract certified randomness from her outputs of the non-local game, i.e.~we need to show that Alice's outputs contain a certain amount of min-entropy conditioned on Eve's side information.
Concretely, in each round of the protocol Alice samples inputs $x$ and $y$ for the non-local game, inputs $x$ into her device to receive outcome $a$, and sends $y$ to Eve to receive outcome $b$; Alice then checks whether $(x,y,a,b)$ satisfies the winning condition of the non-local game.
For comparison, recall that in standard DI randomness expansion~\cite{colbeck2009quantum,Colbeck_2011,pironio2010random,vazirani2012certifiable,miller2016robust}, Alice receives two devices from Eve and uses them to play the non-local game.
This means that in standard DI randomness expansion, Eve never learns any of the inputs and outputs of the game.
In contrast, in blind randomness expansion Eve learns one of the inputs, $y$, and is free to choose one of the outputs, $b$, herself.
Hence, Eve can choose the output $b$ based on past side information and update her side information in each round of the protocol using the values of $y$ and $b$.

To analyse such a protocol, we use the setting of \cref{thm:intro_eat}, with $A_i$ representing the output of Alice's device $D$ from the non-local game in the $i$-th round, $R_i$ the internal memory of $D$ after the $i$-th round, and $E_i$ Eve's side information after the $i$-th round, which can be generated arbitrarily from entanglement shared between Eve and $D$ at the start of the protocol and information Eve gathered during the first $i$ rounds of the protocol.
The map $\cM_i$ describes one round of the protocol, and because Alice's device and Eve cannot communicate during the protocol it is easy to show that the non-signalling condition from \cref{thm:intro_eat} is satisfied.
Therefore, we can apply \cref{thm:intro_eat} to lower-bound Alice's conditional min-entropy $\hmin(A^n|E_n)$ in terms of the single-round quantities $\inf_{\omega} H(A_i|E_i \tilde E_{i-1} )_{\cM_i(\omega)}$.\footnote{In fact, in order for this single-round quantity to be positive one has to restrict the infimum to input states that allow the non-local game to be won with a certain probability. This requires using the generalised EAT with testing (\cref{sec:testing}), not \cref{thm:intro_eat}. We refer to \cref{sec:blind_re} for details.}
This single-round quantity corresponds to the i.i.d.~scenario, i.e.~the generalised EAT has reduced the problem of showing blind randomness expansion against general adversaries to the (much simpler) problem of showing it against i.i.d.~adversaries.
The quantity $\inf_{\omega} H(A_i|E_i \tilde E_{i-1} )_{\cM_i(\omega)}$ can be computed using a general numerical technique~\cite{brown2021computing}, and for certain classes of non-local games it may also be possible to find an analytical lower bound using ideas from~\cite{miller2017randomness,fu2018local}.
Inserting the single-round bound, we obtain a lower bound on $\hmin(A^n|E_n)$ that scales linearly with $n$, showing that blind randomness expansion is possible against general adversaries.
We also note that as explained in~\cite{miller2017randomness}, this result immediately implies that \emph{unbounded} randomness expansion is possible with only three devices, whereas previous works required four devices~\cite{miller2016robust,chung2014physical,coudron_infinite}.

\paragraph{Future work.}
In this work, we have developed a new information-theoretic tool, the generalised EAT.
The generalised EAT deals with a more general model of side information than previous techniques and is therefore more broadly and easily applicable.
In particular, our generalised EAT can be used to analyse mistrustful DI cryptography.
We have demonstrated this by giving the first proof of blind randomness expansion against general adversaries.
We expect that the generalised EAT could similarly be used for other protocols such as two-party cryptography in the noisy storage model~\cite{kaniewski2016device} or certified deletion~\cite{fu2018local,broadbent2020quantum,kundu2020composably}.
In addition to mistrustful DI cryptography, our result can also be used to give new proofs for device-dependent QKD, as demonstrated in \cref{sec:qkd} and~\cite{eat_for_qkd}, and is applicable to proving the security of commercial quantum random number generators, which typically have correlations between rounds due to experimental imperfections~\cite{frauchiger2013true}.

Beyond cryptography, the generalised EAT is useful whenever one is interested in bounding the min-entropy of a large system that can be decomposed in a sequential way.
Such problems are abundant in physics. For example, the dynamics of an open quantum system can be described in terms of interactions that take place sequentially with different parts of the system's  environment~\cite{Campbell2021}.   
In quantum thermodynamics, such a description is commonly employed to model the thermalisation of a system that is brought in contact with a thermal bath. For a lack of techniques, the entropy flow during a thermalisation process of this type is usually quantified in terms of von Neumann entropy rather than the operationally more relevant smooth min- and max-entropies~\cite{delRio2016}.
The generalised EAT may be used to remedy this situation. A similar situation arises in quantum gravity, where smooth entropies play a role in the study of black holes ~\cite{Akers2021}.

In a different direction, one can also try to further improve the generalised EAT itself.
Compared to the original EAT~\cite{eat}, our generalised EAT features a more general model of side information and a weaker condition on the relation between different rounds, replacing the Markov condition of~\cite{eat} with our weaker non-signalling condition in \cref{eqn:intro_ns}.
It is natural to ask whether a further step in this direction is possible: while the model of side information we consider is fully general, it may be possible to replace the non-signalling condition with a weaker requirement.
We have argued above that our non-signalling condition appears to be the most general way of stating the requirement that future side information does not reveal information about past outputs, which seems necessary for an EAT-like theorem.\footnote{In an EAT-like theorem, the entropy contribution from a particular round $i$ has to be calculated conditioned on the side information revealed in that round because we want to analyse the process round-by-round, not globally. If a future round revealed additional side information, then the total entropy contributed by round $i$ would decrease, but there is no way of accounting for that in an EAT-like theorem that simply sums up single-round contributions. As an extreme case, the last round of the process could reveal all prior outputs as side information, so that the total amount of conditional entropy produced by the process is 0, but single-round entropy contributions could be positive. This demonstrates the need for some condition that enforces that future side information does not reveal information about past outputs. We note that this does not mean that there is no way of proving an entropy lower bound in more general settings: for example,~\cite{jain2022direct} do show a bound on the entropy produced by parallel repeated non-local games, but this requires a global analysis.}
It would be interesting to formalise this intuition and see whether our theorem is provably ``tight'' in terms of the conditions placed on the sequential process.
Furthermore, it might be possible to improve the way the statistical condition in~\cref{thm:with_testing} is dealt with in the proof, e.g.~using ideas from~\cite{qpe1,qpe2}.

Finally, one could attempt to extend entropy accumulation from conditional entropies to relative entropies.
Such a \emph{relative entropy accumulation theorem} (REAT) would be the following statement: for two sequences of channels $\{\cE_1, \dots, \cE_n\}$ and $\{\cF_1, \dots, \cF_n\}$ (where $\cF_i$ need not necessarily be trace-preserving), and $\eps > 0$,
\begin{align*}
\dmaxeps{\cE_n \circ \dots \circ \cE_1}{\cF_n \circ \dots \circ \cF_1} \stackrel{?}{\leq} \sum_{i=1}^n \dreg{\cE_i}{\cF_i} + O(\sqrt{n}) \,.
\end{align*}
Here, $D_{\setft{max}}^\eps$ is the $\eps$-smooth max-relative entropy~\cite{tomamichel2015quantum} and we used the (regularised) channel divergences defined in \cref{def:channel_dalpha}.
The key technical challenge in proving this result is to show that the regularised channel divergence $\drega{\cE_i}{\cF_i}$ is continuous in $\alpha$ at $\alpha = 1$, which is an important technical open question.
If one had such a continuity statement and the maps $\cF_i$ additionally satisfied a non-signalling condition (which is not required for the statement above), one could also use our~\cref{thm:improved_chain_rule} to derive a more general REAT, which would imply our generalised EAT.

\paragraph{Acknowledgements} We thank Rotem Arnon-Friedman, Peter Brown, Kun Fang, Raban Iten, Joseph M.~Renes, Ernest Tan, Jinzhao Wang, John Wright, and Yuxiang Yang for helpful discussions and the anonymous FOCS reviewers for useful comments. We further thank Mario Berta and Marco Tomamichel for insights on~\cref{lem:standard_chain_rule}, and Fr\'ed\'eric Dupuis and Carl Miller for discussions about blind randomness expansion. Part of this work was carried out when DS was with the Institute for Theoretical Physics at ETH Zurich. 
TM and RR acknowledge support from the National Centres of Competence in Research (NCCRs) QSIT (funded by the Swiss National Science Foundation under grant number 51NF40-185902) and SwissMAP, the Air Force Office of Scientific Research (AFOSR) via project No. FA9550-19-1-0202, the SNSF project No. 200021\_188541 and the QuantERA project eDICT.
OF acknowledges funding from the European Research Council (ERC Grant Agreement No.~851716).
\newpage

\section{Preliminaries}
\subsection{Notation} \label{sec:notation}
The set of positive semidefinite operators on a quantum system $A$ (with associated Hilbert space $\H_A$) is denoted by $\pos(A)$.
The set of quantum states is given by $\states(A) = \{\rho \in \pos(A) \, | \, \tr{\rho} = 1\}$. 
The set of completely positive maps from linear operators on $A$ to linear operators on $A'$ is denoted by $\cp(A, A')$.
If such a map is additionally trace preserving, we call it a quantum channel and denote the set of such maps by $\cptp(A, A')$.
The identity channel on system $A$ is denoted as $\id_A$.
The spectral norm is denoted by $\norm{\cdot}_{\infty}$.

If $A$ is a quantum system and $X$ is a classical system with alphabet $\cX$, we call $\rho \in S(XA)$ a cq-state and can expand it as 
\begin{align*}
\rho_{XA} = \sum_{x \in \cX} \proj{x} \ot \rho_{A, x}
\end{align*}
for subnormalised $\rho_{A, x} \in \pos(A)$. 
For $\Omega \subset \cX$, we define the conditional state
\begin{align*}
\rho_{XA|\Omega} = \frac{1}{\prs{\rho}{\Omega}} \sum_{x \in \Omega} \proj{x} \ot \rho_{A, x} \,, \quad \text{where } \; \prs{\rho}{\Omega} \deq \sum_{x \in \Omega} \tr{\rho_{A, x}} \,.
\end{align*}
If $\Omega = \{x\}$, we also write $\rho_{XA|x}$ for $\rho_{XA|\Omega}$.

\subsection{R\'enyi divergence and entropy} \label{sec:entropies_prelim}
We will make extensive use of the sandwiched R\'enyi divergence~\cite{MLDSFT13,WWY14} and quantities associated with it, namely R\'enyi entropies and  channel divergences.
We recall the relevant definitions here.
\begin{definition}[R\'enyi divergence] \label{def:dalpha}
For $\rho \in \states(A)$, $\sigma \in \pos(A)$, and $\alpha \in [1/2,1) \cup (1,\infty)$ the (sandwiched) R\'enyi divergence is defined as
\begin{align*}
\dalpha{\rho}{\sigma} \deq \frac{1}{\alpha -1} \log \tr{ \Big(\sigma^{\frac{1-\alpha}{2 \alpha}} \rho \sigma^{\frac{1-\alpha}{2 \alpha}}\Big)^{\alpha}}
\end{align*}
for $\supp(\rho) \subseteq \supp(\sigma)$, and $+\infty$ otherwise.
\end{definition}

From the R\'enyi divergence, one can define the conditional R\'enyi entropies as follows (see~\cite{tomamichel2015quantum} for more details). 
\begin{definition}[Conditional R\'enyi entropy] \label{def:halpha}
For a bipartite state $\rho_{AB} \in \states(AB)$ and $\alpha \in [1/2,1) \cup (1,\infty)$, we define the following two conditional R\'enyi entropies:
\begin{align*}
    \had(A|B)_{\rho}=-D_{\alpha}(\rho_{AB} \| \1_A \otimes \rho_B) \qquad \textnormal{and} \qquad \hau(A|B)_\rho= \sup_{\sigma_B \in \states(B)}-D_{\alpha}(\rho_{AB} \| \1_A \otimes \sigma_B)\, .
\end{align*}
\end{definition}
From the definition it is clear that $\had(A|B) \leq \hau(A|B)$.
Importantly, a relation for the other direction also holds.
\begin{lemma}[{\cite[Corollary 5.3]{tomamichel2015quantum}}] \label{lem:halpha_up_down}
For $\rho_{AB} \in \states(AB)$ and $\alpha \in (1,2)$:
\begin{align*}
\had(A|B)_{\rho} \geq \hauvar{\frac{1}{2-\alpha}}(A|B)_{\rho} \,.
\end{align*}
\end{lemma}

In the limit $\alpha \to 1$ the sandwiched R\'enyi divergence converges to the relative entropy:
\begin{align*}
\lim_{\alpha \to 1} \dalpha{\rho}{\sigma} = D(\rho \| \sigma) = \tr{\rho (\log \rho - \log \sigma)} \,.
\end{align*}
Accordingly, the conditional R\'enyi entropy converges to the conditional von Neumann entropy:
\begin{align*}
\lim_{\alpha \to 1} \had(A|B)_{\rho} = H(A|B)_\rho = H(AB)_\rho - H(B)_\rho = -\tr{\rho_{AB} \log \rho_{AB}} + \tr{\rho_{B} \log \rho_{B}}\,.
\end{align*}

Conversely, in the limit $\alpha \to \infty$, the R\'enyi entropy $\hau$ converges to the min-entropy. 
We will make use of a smoothed version of the min-entropy, which is defined as follows~\cite{renner_thesis}. 
\begin{definition}[Smoothed min-entropy] \label{def:min_entropy}
For $\rho_{AB} \in \states(AB)$ and $\eps \in [0,1]$, the $\eps$-smoothed min-entropy of $A$ conditioned on $B$ is
\begin{align*}
    H_{\min}^\eps(A|B)_{\rho} = - \log \inf_{\tilde \rho_{AB} \in \cB_{\eps}(\rho_{AB})} \inf_{\sigma_{B} \in \states(B)} \norm{\sigma_B^{-\frac{1}{2}} \tilde \rho_{AB} \sigma_B^{-\frac{1}{2}} }_{\infty} \, ,
\end{align*}
where $\norm{\cdot}_{\infty}$ denotes the spectral norm and $\cB_{\eps}(\rho_{AB})$ is the $\eps$-ball around $\rho_{AB}$ in term of the purified dis\-tance~\cite{tomamichel2015quantum}.
\end{definition}

Finally, we can extend the definition of the R\'enyi divergence from states to channels. 
The resulting quantity, the channel divergence (and its regularised version), will play an important role in the rest of the manuscript.
\begin{definition}[Channel divergence] \label{def:channel_dalpha}
For $\cE \in \cptp(A, A')$, $\cF \in \cp(A, A')$, and $\alpha \in [1/2,1) \cup (1,\infty)$, the (stabilised) channel divergence is defined as 
\begin{align} \label{eq:channelDiv}
\dalpha{\cE}{\cF} = \sup_{\omega \in \states(A\tilde A)} \dalpha{\cE(\omega)}{\cF(\omega)} \,,
\end{align}
where without loss of generality $\tilde A \equiv A$.
The regularised channel divergence is defined as
\begin{align*}
\drega{\cE}{\cF} \deq \lim_{n \to \infty} \frac{1}{n} \dalpha{\cE^{\otimes n}}{\cF^{\otimes n}} = \sup_n \frac{1}{n} \dalpha{\cE^{\otimes n}}{\cF^{\otimes n}} \,.
\end{align*}
\end{definition}

We note that the channel divergence is in general not additive under the tensor product~\cite[Proposition~3.1]{FFRS20}, so the regularised channel divergence can be strictly larger that the non-regularised one, i.e., $\drega{\cE}{\cF} > \dalpha{\cE}{\cF}$.
The regularised channel divergence, however, does satisfy an additivity property:
\begin{align}
\drega{\cE^{\otimes k}}{\cF^{\otimes k}} 
= \lim_{n \to \infty} \frac{1}{n} \dalpha{\cE^{\otimes k  n}}{\cF^{\otimes k  n}} 
= k \lim_{n \to \infty} \frac{1}{n'} \dalpha{\cE^{\otimes n'}}{\cF^{\otimes n'}}
= k  \, \drega{\cE}{\cF} \,, \label{eqn:reg_additivity}
\end{align}
where we switched to the index $n' = k n$ for the second equality.

\subsection{Spectral pinching}
A key technical tool in our proof will be the use of spectral pinching maps~\cite{Hayashi_2017}, which are defined as follows (see~\cite[Chapter~3]{Sutter_book} for a more detailed introduction).
\begin{definition}[Spectral pinching map]
Let $\rho \in \pos(A)$ with spectral decomposition $\rho=\sum_{\lambda} \lambda P_{\lambda}$, where $\lambda \in \spec(\rho) \subset \R_{\geq 0}$ are the distinct eigenvalues of $\rho$ and $P_{\lambda}$ are mutually orthogonal projectors. 
The \emph{(spectral) pinching map} $\cP_{\rho} \in \cptp(A, A)$ associated with $\rho$ is given by
\begin{align*}
\cP_{\rho}(\omega) \deq \sum_{\lambda \in \spec(\rho)}  P_\lambda \,  \omega \,  P_{\lambda} \,.
\end{align*}
\end{definition}
We will need a few basic properties of pinching maps.
\begin{lemma}[Properties of pinching maps] \label{lem:pinching_property}
For any $\rho,\sigma \in \pos(A)$, the following properties hold:
\begin{enumerate}
    \item  Invariance: $\cP_{\rho}(\rho)=\rho$\,.
    \item  Commutation of pinched state: $[\sigma, \cP_{\sigma}(\rho)] = 0$\,.
    \item  Pinching inequality: \smash{$\cP_{\sigma}(\rho) \geq \frac{1}{|\spec(\sigma)|}\rho$}\,. 
    \item Commutation of pinching maps: if $[\rho,\sigma]=0$, then $\cP_{\rho} \circ \cP_{\sigma} = \cP_{\sigma} \circ \cP_{\rho}$\,.
    \item Partial trace: $\ptr{B}{\cP_{\rho_A \otimes \1_B}(\omega_{AB})} = \cP_{\rho_A}(\omega_A) \quad \forall\,\omega_{AB} \in \pos(AB)$.
\end{enumerate}
\end{lemma}
\begin{proof}
Properties (i)--(iii) follow from the definition and~\cite[Chapter 2.6.3]{tomamichel09} or \cite[Lemma~3.5]{Sutter_book}.

For the fourth statement, note that since $[\rho, \sigma] = 0$, there exists a joint orthonormal eigenbasis $\{\ket{x_i}\}$ of $\rho$ and $\sigma$.
Let $P_\lambda$ be the projector onto the eigenspace of $\rho$ with eigenvalue $\lambda$, and $Q_\mu$ the projector onto the eigenspace of $\sigma$ with eigenvalue $\mu$.
We can expand 
\begin{align*}
P_\lambda = \sum_{i \sth \rho \ket{x_i} = \lambda \ket{x_i}} \proj{x_i}  \qquad \textnormal{and} \qquad Q_\mu = \sum_{j \sth \sigma \ket{x_j} = \mu \ket{x_j}} \proj{x_j} \,.
\end{align*}
Since $\{\ket{x_i}\}$ is a family of orthonormal vectors, 
\begin{align*}
P_\lambda Q_{\mu} = \sum_{\substack{i \sth \rho \ket{x_i} = \lambda \ket{x_i} \\ \tand \sigma \ket{x_i} = \mu \ket{x_i}}} \proj{x_i} = Q_{\mu} P_\lambda \,,
\end{align*}
which implies commutation of the pinching maps.

For the fifth statement, note that if we write $\rho = \sum_{\lambda} \lambda P_\lambda$ with eigenprojectors $P_{\lambda}$, then the set of eigenprojectors of $\rho_A \ot \1_B$ is simply $\{P_\lambda \ot \1_B\}$.
Hence, 
\begin{align*}
\ptr{B}{\cP_{\rho_A \otimes \1_B}(\omega_{AB})} = \sum_{\lambda} \ptr{B}{P_\lambda \ot \1_B \omega_{AB} P_\lambda \ot \1_B} = \sum_{\lambda} P_\lambda \ptr{B}{\omega_{AB}} P_\lambda = \cP_{\rho_A}(\omega_A) \,.
\end{align*}
\end{proof}

It is often useful to use the pinching map associated with tensor power states, i.e.,  $\cP_{\rho^{\otimes n}}$.
This is because for $\rho \in \pos(A)$, the factor $|\spec(\rho^{\otimes n})|$ from the pinching inequality (see~\cref{lem:pinching_property}) only scales polynomially in $n$ (see e.g.~\cite[Remark~3.9]{Sutter_book}):
\begin{align}
|\spec(\rho^{\otimes n})| \leq (n+1)^{\dim(A)-1} \,. \label{eqn:tp_spec_bound}
\end{align}
In fact, we can show a similar property for all permutation-invariant states, not just tensor product states.

\begin{lemma} \label{lem:permutation_inv_small_spectrum}
Let $\rho \in \pos(A^{\ot n})$ be permutation invariant and denote $d = \dim(A)$. Then 
\begin{align*}
|\spec(\rho)| \leq (n+d)^{d(d+1)/2} \,.
\end{align*}
\end{lemma}
\begin{proof}
By Schur-Weyl duality, since $\rho$ is permutation-invariant, we have
\begin{align*}
\rho \cong \bigoplus_{\lambda \in \cI_{d, n}} \rho(\lambda)_{Q_\lambda} \otimes \1_{P_{\lambda}} \,,
\end{align*}
where $\cong$ denotes equality up to unitary conjugation (which leaves the spectrum invariant), $\cI_{d, n}$ is the set of Young diagrams with $n$ boxes and at most $d$ rows, $Q_\lambda$ and $P_\lambda$ are systems whose details need not concern us, and $\rho(\lambda) \in \pos(Q_\lambda)$.
From this it is clear that 
\begin{align*}
|\spec(\rho)| \leq \sum_{\lambda \in \cI_{d, n}} |\spec(\rho(\lambda))| \leq \sum_{\lambda \in \cI_{d, n}} \dim(Q_{\lambda}) \,.
\end{align*}
It is known that $|\cI_{d, n}| \leq (n+1)^d$ and $\dim(Q_\lambda) \leq (n+d)^{d(d-1)/2}$ (see e.g.~\cite[Section 6.2]{harrow_thesis}).
Hence 
\begin{align*}
|\spec(\rho)| \leq (n+1)^d  (n+d)^{d(d-1)/2} \leq (n+d)^{d(d+1)/2} \,.
\end{align*}
\end{proof}

\begin{corollary} \label{lem:pinching_output_spectrum}
Let $\rho, \sigma \in \pos(A)$ and $d = \dim(A)$. Then 
\begin{align*}
|\spec\!\left( \cP_{\rho^{\otimes n}}(\sigma^{\otimes n}) \right)| \leq (n+d)^{d(d+1)/2} \,.
\end{align*}
\end{corollary}

\begin{proof}
Note that $\cP_{\rho^{\otimes n}}(\sigma^{\otimes n})$ is itself not a product state because the eigenprojectors of $\rho^{\ot n}$ do not have a product form.
However, since every eigenspace of $\rho^{\otimes n}$ is permutation-invariant, $\cP_{\rho^{\otimes n}}(\sigma^{\otimes n})$ is permutation-invariant, too, so we can apply \cref{lem:permutation_inv_small_spectrum}.
\end{proof}

\section{Strengthened chain rules} \label{sec:chain_rule}

One of the crucial properties of entropies are chain rules, which allow us to relate entropies of large composite systems to sums of entropies of the individual subsystems.
In this section, we prove two new such chain rules, one for the R\'enyi divergence (\cref{thm:improved_chain_rule}, which is a generalisation of~\cite[Corollary 5.1]{sharp_divergence}) and one for the conditional entropy (\cref{thm:improved_entropy_chain_rule}).
The chain rule from \cref{thm:improved_entropy_chain_rule} is the key ingredient for our generalised EAT, to which we will turn our attention in \cref{sec:eat}.
\cref{thm:improved_entropy_chain_rule} plays a similar role for our generalised EAT as \cite[Corollary 3.5]{eat} does for the original EAT, but while the latter requires a Markov condition, the former does not.
As a result, our generalised EAT based on \cref{thm:improved_entropy_chain_rule} also avoids the Markov condition.

The outline of this section is as follows: we first prove a generalised chain rule for the R\'enyi divergence (\cref{thm:improved_chain_rule}).
This chain rule contains a regularised channel divergence.
As the next step, we show that in the special case of conditional entropies, we can drop the regularisation (\cref{sec:removing_reg}).
This allows us to derive a chain rule for conditional entropies from the chain rule for channels (\cref{sec:entropy_chain_rule}).

\subsection{Strengthened chain rule for R\'enyi divergence} \label{sec:chain_rule_divergence}
The main result of this section is the following chain rule for the R\'enyi divergence.
\begin{theorem} \label{thm:improved_chain_rule}
Let $\alpha > 1$, $\rho \in \states(A R)$, $\sigma \in \pos(A R)$, $\cE \in \cptp(AR, B)$, and $\cF \in \cp(AR, B)$. Suppose that there exists $\cR \in \cp(A, B)$ such that $\cF = \cR \circ \setft{Tr}_{R}$. 
Then 
\begin{align}
\dalpha{\cE(\rho_{AR})}{\cF(\sigma_{AR})} \leq \dalpha{\rho_{A}}{\sigma_{A}} + \drega{\cE}{\cF} \,. \label{eqn:improved_chain_rule}
\end{align}
\end{theorem}

This is a stronger version of an existing chain rule due to~\cite{sharp_divergence}, which we will use in our proof of \cref{thm:improved_chain_rule}:
\begin{lemma}[{\cite[Corollary 5.1]{sharp_divergence}}] \label{lem:standard_chain_rule}
Let $\alpha > 1$, $\rho \in \states(A)$, $\sigma \in \pos(A)$, $\cE \in \cptp(A, B)$, and $\cF \in \cp(A, B)$. Then 
\begin{align}
\dalpha{\cE(\rho)}{\cF(\sigma)} \leq \dalpha{\rho}{\sigma} + \drega{\cE}{\cF} \,. \label{eqn:standard_chain_rule}
\end{align}
\end{lemma}

The difference between \cref{thm:improved_chain_rule} and \cref{lem:standard_chain_rule} is that on the r.h.s.~of \cref{eqn:improved_chain_rule}, we only have the divergence $\dalpha{\rho_{A}}{\sigma_{A}}$ between the two reduced states on system $A$.
In contrast, if we used \cref{eqn:standard_chain_rule} with systems $AR$, then we would get the divergence $\dalpha{\rho_{AR}}{\sigma_{AR}}$ between the full states.
In particular, the weaker \cref{lem:standard_chain_rule} can easily be recovered from \cref{thm:improved_chain_rule} by taking the system $R$ to be trivial, in which case the condition $\cF = \cR \circ \setft{Tr}_{R}$ becomes trivial, too.

While the difference between~\cref{thm:improved_chain_rule} and~\cref{lem:standard_chain_rule} may look minor at first sight, the two chain rules can give considerably different results:
in general, the data processing inequality ensures that $\dalpha{\rho_{A}}{\sigma_{A}} \leq \dalpha{\rho_{AR}}{\sigma_{AR}}$, but the gap between the two quantities can be significant, i.e., there exist states for which $\dalpha{\rho_{A}}{\sigma_{A}} \ll \dalpha{\rho_{AR}}{\sigma_{AR}}$.
In such cases,~\cref{thm:improved_chain_rule} yields a significantly tighter bound.
This turns out to be crucial if we want to apply this chain rule repeatedly to get an EAT.

We also note that the statement of~\cref{thm:improved_chain_rule} is known to be correct also for $\alpha=1$~\cite[Theorem~3.5]{FFRS20}. 
However, this requires a separate proof and does not follow from~\cref{thm:improved_chain_rule} as it is currently not known whether the function $\alpha \mapsto \drega{\cE}{\cF} $ is continuous in the limit $\alpha \searrow 1$.\footnote{It is well-known~\cite[Lemma~8]{tomamichel09} that $\lim_{\alpha \searrow 1} \dalpha{\cE}{\cF} = D\big(\cE\|\cF\big)$, but it is unclear whether the same holds for the regularised quantity.}

We now turn to the proof of \cref{thm:improved_chain_rule}. 
The key question for the proof is the following: given states $\rho_{AR}$ and $\sigma_{A}$, does there exist an extension $\sigma_{AR}$ of $\sigma_A$ such that $\dalpha{\rho_A}{\sigma_A} = \dalpha{\rho_{AR}}{\sigma_{AR}}$?
For the special case of $\alpha = 1/2$, an affirmative answer is given by Uhlmann's theorem~\cite{Uhlmann76} (see also~\cite[Corollary~3.14]{tomamichel2015quantum}). This also holds for $\alpha = \infty$, but not in general for $\alpha \geq 1$ as discussed in~\cref{sec:uhlmann_alpha}. The following lemma shows that a similar property still holds for $\alpha > 1$ on a regularised level.
\begin{lemma} \label{lem:min_over_ref_system}
Consider quantum systems $A$ and $R$ with $d = \dim(A)$.
For $n \in \N$, we define $A^n = A_1 \dots A_n$, where $A_i$ are copies of the system $A$, and likewise $R^n = R_1 \dots R_n$. 
Then for $\rho \in \states(AR)$, $\sigma \in \pos(A)$, and $\alpha > 1$ we have
\begin{align*}
\dalpha{\rho_A}{\sigma_A} \leq \inf_{\hat \sigma_{A^n R^n} \sth \hat \sigma_{A^n} = \sigma_A^{\otimes n}} \frac{1}{n} \dalpha{\rho_{AR}^{\otimes n}}{\hat \sigma_{A^n R^n}} \leq \dalpha{\rho_A}{\sigma_A} + \frac{\alpha}{\alpha-1} \frac{d(d+1) \log(n+d)}{n}  \, .
\end{align*}
\end{lemma}
\begin{proof}
The inequality 
\begin{align*}
\dalpha{\rho_A}{\sigma_A} \leq \inf_{\hat \sigma_{A^n R^n} \sth \hat \sigma_{A^n} = \sigma_A^{\otimes n}} \frac{1}{n} \dalpha{\rho_{AR}^{\otimes n}}{\hat \sigma_{A^n R^n}}
\end{align*}
follows directly from the data processing inequality for taking the partial trace over $R^n$, and additivity of $D_{\alpha}$ under tensor product~\cite{tomamichel2015quantum}.

For the other direction, we consider $n$-fold tensor copies of $\rho_{AR}$ and $\sigma_{A}$, which we denote by $\rho_{A^n R^n} = \rho_{A_1 R_1} \otimes \dots \otimes \rho_{A_n R_n}$ and  $\sigma_{A^n} = \sigma_{A_1} \otimes \dots \otimes \sigma_{A_n}$.
We define the following two pinched states 
\begin{align} \label{eqn:def_rho_prime}
\rho'_{A^n R^n} = \cP_{\sigma_{A^n} \otimes \1_{R^n}}(\rho_{A^n R^n}) \qquad \textnormal{and} \qquad 
\hat\rho_{A^n R^n} = \cP_{\rho'_{A^n} \otimes \1_{R^n}}(\rho'_{A^n R^n}) \,.
\end{align}
By definition of $\hat\rho_{A^n R^n}$ and using the pinching inequality (see~\cref{lem:pinching_property}(iii)) twice, we have 
\begin{align*}
\rho_{A^n R^n} \leq |\spec(\sigma_{A^n})| |\spec(\rho'_{A^n})| \; \hat\rho_{A^n R^n} \,. 
\end{align*}
Using the operator monotonicity of the sandwiched R\'enyi divergence in the first argument~\cite{tomamichel2015quantum} we find for any state $\hat \sigma_{A^n R^n}$
\begin{align}
\frac{1}{n} \dalpha{\rho_{AR}^{\otimes n}}{\hat \sigma_{A^nR^n}} \leq \frac{1}{n} \dalpha{\hat \rho_{A^nR^n}}{\hat \sigma_{A^nR^n}} + \frac{1}{n} \frac{\alpha}{\alpha-1} \eta(n) \,, \label{eqn:pinched_d}
\end{align}
with the error term 
\begin{align*}
\eta(n) = \log |\spec(\sigma_{A^n})| + \log |\spec(\rho'_{A^n})| \,.
\end{align*}

To prove the lemma, we now need to bound the error term $\eta(n)$ and construct a specific choice for $\hat \sigma_{A^nR^n}$ for which $\hat \sigma_{A^n} = \sigma_A^{\otimes n}$ and $\frac{1}{n} \dalpha{\hat \rho_{A^nR^n}}{\hat \sigma_{A^nR^n}} \leq \dalpha{\rho_A}{\sigma_A}$.
We first bound $\eta(n)$.
Since $\sigma_{A^n} = \sigma_A^{\otimes n}$, we have from \cref{eqn:tp_spec_bound} that $|\spec(\sigma_{A^n})| \leq (n+1)^{d-1}$, where $d = \dim(A)$.
To bound $|\spec(\rho'_{A^n})|$, we note that by~\cref{eqn:def_rho_prime} and~\cref{lem:pinching_property}(v)
\begin{align}
\rho'_{A^n} = \ptr{R^n}{\cP_{\sigma_{A^n} \otimes \1_{R^n}}(\rho_{A^n R^n})} = \cP_{\sigma_{A^n}}(\rho_{A^n}) = \cP_{\sigma_{A}^{\otimes n}}(\rho_{A}^{\otimes n}) \,.
\end{align}
We can therefore use \cref{lem:pinching_output_spectrum} to obtain $|\spec(\rho'_{A^n})| \leq (n+d)^{d(d+1)/2}$.
Hence, 
\begin{align}
\eta(n) \leq d(d+1) \log(n+d) \,. \label{eqn:bound_eta}
\end{align}

It thus remains to construct $\hat \sigma_{A^n R^n}$ satisfying the properties mentioned above.
To do so we first establish a number of commutation statements.
\begin{enumerate}
\item From \cref{lem:pinching_property}(ii) we have that $[\hat\rho_{A^n R^n}, \rho'_{A^n} \otimes \1_{R^n}] = 0$.
Recalling the definition of $\rho'$ from~\cref{eqn:def_rho_prime}, we get
\begin{align}
\hat{\rho}_{A^n} = \ptr{R^n}{\cP_{\rho'_{A^n} \otimes \1_{R^n}}(\rho'_{A^n R^n})} = \cP_{\rho'_{A^n}} (\rho'_{A^n}) = \rho'_{A^n} \,,\label{eqn:hat_prime_equal}
\end{align}
where the final step uses~\cref{lem:pinching_property}(i). As a result we find
\begin{align}
[\hat\rho_{A^n R^n}, \hat\rho_{A^n} \otimes \1_{R^n}] = 0 \,. \label{eqn:comm1}
\end{align}
\item From  \cref{lem:pinching_property}(ii) we have that $[\rho'_{A^n R^n}, \sigma_{A^n} \otimes \1_{R^n}] = 0$.
Taking the partial trace over $R^n$, this implies $[\rho'_{A^n}, \sigma_{A^n}] = 0$, so by~\cref{lem:pinching_property}(iv) and~\cref{eqn:def_rho_prime} 
\begin{align*}
\hat\rho_{A^n R^n} = \cP_{\rho'_{A^n} \otimes \1_{R^n}}\left( \cP_{\sigma_{A^n} \otimes \1_{R^n}}(\rho_{A^n R^n}) \right) = \cP_{\sigma_{A^n} \otimes \1_{R^n}}\left( \cP_{\rho'_{A^n} \otimes \1_{R^n}}(\rho_{A^n R^n}) \right) \,.
\end{align*}
Therefore, by \cref{lem:pinching_property}(ii), 
\begin{align}
[\hat\rho_{A^n R^n}, \sigma_{A^n} \otimes \1_{R^n}] = 0 \,. \label{eqn:comm2}
\end{align}
\item Taking the partial trace over $R^n$ in \cref{eqn:comm2}, we get 
\begin{align}
[\hat\rho_{A^n}, \sigma_{A^n}] = 0 \,. \label{eqn:comm3}
\end{align}
\end{enumerate}
Having established these commutation relations, we define $\cT \in \cptp(A^n, A^n R^n)$ by\footnote{In case $\hat\rho_{A^n}$ does not have full support, we only take the inverse on the support of $\hat\rho_{A^n}$.} 
\begin{align*}
\cT(\omega_{A^n}) = \hat\rho_{A^n R^n}^{1/2} \hat\rho_{A^n}^{-1/2} \omega_{A^n} \hat\rho_{A^n}^{-1/2} \hat\rho_{A^n R^n}^{1/2}\,.
\end{align*}
By construction, 
\begin{align}
\cT(\hat\rho_{A^n}) = \hat\rho_{A^n R^n} \,. \label{eqn:rho_with_T}
\end{align}
We define 
\begin{align}
\hat{\sigma}_{A^n R^n} = \cT(\sigma_{A^n}) \,. \label{eqn:sigma_with_T}
\end{align}
To see that this is a valid choice of $\hat \sigma$, i.e., that $\hat{\sigma}_{A^n} = \sigma_{A^n} = \sigma_A^{\otimes n}$, we use \cref{eqn:comm1}, \cref{eqn:comm2}, and \cref{eqn:comm3} to find 
\begin{align*}
\hat{\sigma}_{A^n} = \ptr{R^n}{\hat\rho_{A^n R^n}^{1/2} \hat\rho_{A^n}^{-1/2} \sigma_{A^n} \hat\rho_{A^n}^{-1/2} \hat\rho_{A^n R^n}^{1/2}} = \ptr{R^n}{\hat\rho_{A^n R^n} \hat\rho_{A^n}^{-1} \sigma_{A^n}} = \sigma_{A^n} \,.
\end{align*}

Using \cref{eqn:rho_with_T} and \cref{eqn:sigma_with_T} followed by the data processing inequality~\cite{tomamichel2015quantum}, we obtain
\begin{align}
\frac{1}{n} \dalpha{\hat \rho_{A^n R^n}}{\hat \sigma_{A^n R^n}} 
= \frac{1}{n} \dalpha{\cT(\hat\rho_{A^n})}{\cT(\sigma_{A^n})}
\leq \frac{1}{n} \dalpha{\hat\rho_{A^n}}{\sigma_{A^n}} \,. \label{eqn:chain_rule2}
\end{align}
By \cref{eqn:hat_prime_equal} and \cref{eqn:def_rho_prime} we have $\hat\rho_{A^n} = \rho'_{A^n} = \cP_{\sigma_{A^n}}(\rho_{A^n})$.
Therefore, continuing from \cref{eqn:chain_rule2} and using $\sigma_{A^n} = \cP_{\sigma_{A^n}}(\sigma_{A^n})$ followed by the data processing inequality gives 
\begin{align*}
\frac{1}{n} \dalpha{\hat \rho_{A^n R^n}}{\hat \sigma_{A^n R^n}} \leq  \frac{1}{n} \dalpha{\rho_{A^n}}{\sigma_{A^n}} = \frac{1}{n} \dalpha{\rho_{A}^{\otimes n}}{\sigma_{A}^{\otimes n}} = \dalpha{\rho_{A}}{\sigma_{A}} \,.
\end{align*}
Inserting this and our error bound from \cref{eqn:bound_eta} into \cref{eqn:pinched_d} proves the desired statement.
\end{proof}

With this, we can now prove \cref{thm:improved_chain_rule}.

\begin{proof}[Proof of \cref{thm:improved_chain_rule}]
Because $D_{\alpha}$ is additive under tensor products, for any $n \in \N$ we have
\begin{align}
\dalpha{\cE(\rho_{AR})}{\cF(\sigma_{AR})} 
&= \frac{1}{n} \dalpha{\cE^{\otimes n}(\rho_{AR}^{\otimes n})}{\cF^{\otimes n}(\sigma_{AR}^{\otimes n})} \nonumber \\
&= \inf_{\hat \sigma_{A^n R^n} \sth \hat \sigma_{A^n} = \sigma_A^{\otimes n}} \frac{1}{n} \dalpha{\cE^{\otimes n}(\rho_{AR}^{\otimes n})}{\cF^{\otimes n}(\hat \sigma_{A^nR^n})}\,, \label{eqn:min_over_sigma}
\end{align}
where the second equality holds because $\cF = \cR \circ \setft{Tr}_{R}$, so $\cF^{\otimes n}(\sigma_{AR}^{\otimes n}) = \cF^{\otimes n}(\hat \sigma_{A^nR^n})$ for any $\hat \sigma_{A^nR^n}$ that satisfies $\hat \sigma_{A^n} = \sigma_A^{\otimes n}$.
From the chain rule in \cref{lem:standard_chain_rule} we get that for any $\hat \sigma_{A^nR^n}$:
\begin{align*}
\frac{1}{n} \dalpha{\cE^{\otimes n}(\rho_{AR}^{\otimes n})}{\cF^{\otimes n}(\hat \sigma_{A^nR^n})} 
& \leq \frac{1}{n} \dalpha{\rho_{AR}^{\otimes n}}{\hat \sigma_{A^n R^n}} + \frac{1}{n} \drega{\cE^{\otimes n}}{\cF^{\otimes n}} \\
&= \frac{1}{n} \dalpha{\rho_{AR}^{\otimes n}}{\hat \sigma_{A^n R^n}} + \drega{\cE}{\cF} \,,
\end{align*}
where for the second line we used additivity of the regularised channel divergence (see~\cref{eqn:reg_additivity}).
Combining this with \cref{eqn:min_over_sigma}, we get 
\begin{align*}
\dalpha{\cE(\rho_{AR})}{\cF(\sigma_{AR})} \leq \inf_{\hat \sigma_{A^n R^n} \sth \hat \sigma_{A^n} = \sigma_A^{\otimes n}} \frac{1}{n} \dalpha{\rho_{AR}^{\otimes n}}{\hat \sigma_{A^nR^n}} + \drega{\cE}{\cF} \,.
\end{align*}
Taking $n\to\infty$, the theorem then follows from \cref{lem:min_over_ref_system}.
\end{proof}

\subsection{Removing the regularisation} \label{sec:removing_reg}
The chain rule presented in~\cref{thm:improved_chain_rule} contains a regularised channel divergence term, which cannot be computed easily and whose behaviour as $\alpha \searrow 1$ is not understood. 
In this section we show that in the specific case relevant for entropy accumulation, this regularisation can be removed.
From this, we then derive a chain rule for R\'enyi entropies in \cref{thm:improved_entropy_chain_rule}.

\begin{definition}[Replacer map]
The replacer map $\cS_{A} \in \cp(A,A)$ is defined by its action on an arbitrary state $\omega_{AR}$: 
\begin{align*}
\cS_A(\omega_{AR}) = \1_A \otimes \omega_{R} \,.
\end{align*}
\end{definition}

\begin{lemma} \label{lem:entropy_reg}
Let $\alpha \in (1,2)$, $\cE \in \cptp(AR, A'R')$, and $\cF = \cS_{A'} \circ \cE$, where $\cS_{A'}$ is the replacer map. Then we have 
\begin{align*}
\drega{\cE}{\cF} \leq \dvaralpha{\frac{1}{2 - \alpha}}{\cE}{\cF} \,.
\end{align*}
\end{lemma}
\begin{proof}
Due to the choice of $\cF$, we have that for any state $\psi^n \in \states(A^nR^n \tilde R^n)$ (with $\tilde R \equiv AR$):
\begin{align*}
\dalpha{\cE^{\otimes n}(\psi^n)}{\cF^{\otimes n}(\psi^n)} = - \had\left((A')^{n} | (R')^{n} \tilde R^n \right)_{\cE^{\otimes n}(\psi^n)} \,.
\end{align*}
From \cite[Proposition II.4]{leditzky2018approaches} and \cite[Lemma 4.2.2]{renner_thesis} we know that for every $n$, there exists a symmetric pure state $\ket{\hat\psi^n} \in \sym^n(A R \tilde R)$ such that 
\begin{align*}
\dalpha{\cE^{\otimes n}}{\cF^{\otimes n}} = \dalpha{\cE^{\otimes n}(\hat\psi^n)}{\cF^{\otimes n}(\hat \psi^n)} = - \had\left((A')^{n} | (R')^{n} \tilde R^n \right)_{\cE^{\otimes n}(\psi^n)} \,,
\end{align*}
where $\hat\psi^n = \proj{\hat \psi^n}$ and the supremum in the definition of the channel divergence is achieved because the conditional entropy is continuous in the state.
Let $d = \dim(AR\tilde R)$ and $g_{n, d} = \dim(\sym^n(AR\tilde R)) \leq (n+1)^{d^2 - 1}$.
We define the state
\begin{align}
\tau^n_{A^nR^n\tilde R^n} = \int \mu(\sigma_{AR\tilde R}) \sigma_{AR\tilde R}^{\otimes n} \,, \label{eqn:symm_opt}
\end{align}
where $\mu$ is the Haar measure on pure states.
We now claim that in the limit $n\to \infty$, we can essentially replace the optimizer $\hat \psi^n_{A^nR^n\tilde R^n}$ by the state $\tau^n_{A^nR^n\tilde R^n}$ in \cref{eqn:symm_opt}. 
More precisely, we claim that
\begin{align}
\lim_{n \to \infty} \frac{1}{n} \had((A')^{n} | (R')^{n} \tilde R^n)_{\cE^{\otimes n}(\hat \psi^n)} 
\geq \lim_{n \to \infty} \frac{1}{n} \hadvar{\frac{1}{2-\alpha}}((A')^{n} | (R')^{n} \tilde R^n)_{\cE^{\otimes n}(\tau^n)} \,. \label{eqn:entropy_psi_tau}
\end{align}

To show this, we first use \cref{lem:halpha_up_down} to get 
\begin{align*}
\had((A')^{n} | (R')^{n} \tilde R^n)_{\cE^{\otimes n}(\hat \psi^n)} \geq \hauvar{\frac{1}{2 - \alpha}}\left((A')^{n} | (R')^{n} \tilde R^n\right)_{\cE^{\otimes n}(\hat \psi^n)} \,.
\end{align*}
It is know that $\tau^n_{A^nR^n\tilde R^n}$ is the maximally mixed state on $\sym^n(AR\tilde R)$ (see e.g.~\cite{christandl2009postselection}).
Therefore,  
\begin{align*}
\rho^n_{A^nR^n\tilde R^n} \deq \frac{g_{n, d} \tau^n - \hat \psi^n}{g_{n, d} - 1}
\end{align*}
is a valid quantum state (i.e. positive and normalised).
Hence, we can write 
\begin{align*}
\tau^n = \left( 1 - \frac{1}{g_{n, d}} \right) \rho^n + \frac{1}{g_{n, d}} \hat \psi^n \,.
\end{align*}
Using \cite[Lemma B.5]{eat}, it follows that 
\begin{align*}
\frac{1}{n}\hauvar{\frac{1}{2 - \alpha}}\left((A')^{n} | (R')^{n}  \tilde R^n\right)_{\cE^{\otimes n}(\hat \psi^n)} 
\geq \frac{1}{n} \hauvar{\frac{1}{2 - \alpha}}\left((A')^{n} | (R')^{n} \tilde R^n\right)_{\cE^{\otimes n}(\tau^n)} - \frac{\alpha}{\alpha-1} \frac{\log(g_{n, d})}{n} \,.
\end{align*}
Since $\frac{\log(g_{n, d})}{n} \leq (d^2 - 1) \frac{\log n}{n}$ vanishes as $n\to\infty$, taking the limit and using $\hauvar{\frac{1}{2-\alpha}}(\cdot | \cdot) \geq \hadvar{\frac{1}{2-\alpha}}(\cdot | \cdot)$ proves \cref{eqn:entropy_psi_tau}.

Having established \cref{eqn:entropy_psi_tau}, we can now conclude the proof of the lemma as follows
\begin{align*}
\drega{\cE}{\cF} 
&= - \lim_{n \to \infty} \frac{1}{n} \had((A')^{n} | (R')^{n} \tilde R^n)_{\cE^{\otimes n}(\hat \psi^n)} \\
&\leq - \lim_{n \to \infty} \frac{1}{n} \hadvar{\frac{1}{2-\alpha}}((A')^{n} | (R')^{n} \tilde R^n)_{\cE^{\otimes n}(\tau^n)} \\
&= \lim_{n \to \infty} \frac{1}{n} \dvaralpha{\frac{1}{2-\alpha}}{\cE^{\otimes n}\Big( \int \mu(\sigma_{AR\tilde R}) \sigma_{AR\tilde R}^{\otimes n} \Big)}{\cF^{\otimes n}\Big( \int \mu(\sigma_{AR\tilde R}) \sigma_{AR\tilde R}^{\otimes n} \Big)} \\
&\leq \lim_{n \to \infty} \sup_{\sigma_{AR\tilde R} \in \states(AR\tilde R)} \frac{1}{n} \dvaralpha{\frac{1}{2-\alpha}}{\cE^{\otimes n}\left( \sigma_{AR\tilde R}^{\otimes n} \right)}{\cF^{\otimes n}\left( \sigma_{AR\tilde R}^{\otimes n} \right)} \\
&= \dvaralpha{\frac{1}{2-\alpha}}{\cE}{\cF} \,,
\end{align*}
where we used joint quasi-convexity~\cite[Proposition~4.17]{tomamichel2015quantum} in the fourth line and additivity under tensor products in the last line.
\end{proof}

\subsection{Strengthened chain rule for conditional R\'enyi entropy} \label{sec:entropy_chain_rule}
We next combine~\cref{thm:improved_chain_rule} with~\cref{lem:entropy_reg} to derive a new chain rule for the conditional R\'enyi entropy which then allows us to prove the generalised EAT in~\cref{sec:eat}.
\begin{lemma} \label{thm:improved_entropy_chain_rule}
Let $\alpha \in (1, 2)$, $\rho \in \states(ARE)$, and $\cM \in \cptp(RE, A'R'E')$ such that there exists $\cR \in \cptp(E, E')$ such that $\setft{Tr}_{A'R'} \circ \cM = \cR \circ \setft{Tr}_{R}$. 
Then 
\begin{align}
\had(AA'|E')_{\cM(\rho)} \geq \had(A|E)_{\rho} + \inf_{\omega \in \states(RE\tilde E)}\hadvar{\frac{1}{2-\alpha}}(A'|E' \tilde E)_{\cM(\omega)}  \label{eqn:entropy_chain_rule}
\end{align}
for a purifying system $\tilde E \equiv RE$.
\end{lemma}

\begin{proof}
We define the following maps\footnote{The map $\cM$ in the theorem statement is also implicitly tensored with an identity map on $A$, but for the definition of $\tilde M$ we make this explicit to avoid confusion when applying~\cref{thm:improved_chain_rule}.}
\begin{alignat*}{2}
\cN &= \cS_{A'} \circ \cM &&\in \cp(RE, A'R'E')\,,\\
\tilde \cM &= \id_{A} \otimes \setft{Tr}_{R'} \circ \cM &&\in \cptp(ARE, AA'E')\,,\\ 
\tilde \cN &= \cS_{A'} \circ \tilde\cM &&\in \cp(ARE, AA'E') \,.
\end{alignat*}
Note that in \cref{eqn:entropy_chain_rule}, we can replace $\cM$ by $\tilde \cM$, as the system $R'$ does not appear in \cref{eqn:entropy_chain_rule}.
With $\sigma_{ARE} = \1_A \otimes \rho_{RE}$ and $\tilde \cN = \cS_{A'} \circ \tilde \cM$, we can write 
\begin{align*}
- \had(AA'|E')_{\cM(\rho)} = \dalpha{\tilde \cM(\rho_{ARE})}{\tilde \cN(\sigma_{ARE})} \,.
\end{align*}
We now claim that there exists a map $\tilde \cR \in \cp(AE, AA'E)$ such that $\tilde \cN = \tilde \cR \circ \setft{Tr}_{R}$.
To see this, observe that by assumption, $\setft{Tr}_{A'} \circ \tilde \cM = \id_{A} \otimes \cR \circ \setft{Tr}_{R}$ for some $\cR \in \cp(E, E')$.
Then, we can define $\tilde \cR \in \cp(AE, AA'E)$ by its action on an arbitrary state $\omega_{AE}$:
\begin{align*}
\tilde \cR(\omega_{AE}) \deq \1_{A'} \otimes (\id_A \otimes \cR)(\omega_{AE}) = \1_{A'} \ot \setft{Tr}_{A'} \circ \tilde\cM (\omega_{ARE}) = \tilde \cN (\omega_{ARE})
\end{align*} 
for any extension $\omega_{ARE}$ of $\omega_{AE}$.
Therefore, we can apply \cref{thm:improved_chain_rule} to find 
\begin{align*}
\dalpha{\tilde \cM(\rho_{ARE})}{\tilde \cN(\sigma_{ARE})} \leq \dalpha{\rho_{AE}}{\sigma_{AE}} + \drega{\tilde \cM}{\tilde \cN} \,.
\end{align*}
By definition of $\sigma$, we have $\dalpha{\rho_{AE}}{\sigma_{AE}} = - \had(A|E)_{\rho}$. Since the channel divergence is stabilised, tensoring with $\id_{A}$ has no effect, i.e.,
\begin{align*}
\drega{\tilde \cM}{\tilde \cN} = \drega{\setft{Tr}_{R'} \circ \cM}{\setft{Tr}_{R'} \circ \cN} = \drega{\setft{Tr}_{R'} \circ \cM}{\cS_{A'} \circ \setft{Tr}_{R'} \circ \cM} \,.
\end{align*}
To this, we can apply \cref{lem:entropy_reg} and obtain
\begin{align*}
\drega{\tilde \cM}{\tilde \cN} \leq \dvaralpha{\frac{1}{2 - \alpha}}{\setft{Tr}_{R'} \circ \cM}{\cS_{A'} \circ \setft{Tr}_{R'} \circ \cM} = - \inf_{\omega \in \states(RE\tilde E)} \hadvar{\frac{1}{2-\alpha}}(A' | E'\tilde E)_{\cM(\omega)} 
\end{align*}
with $\tilde E \equiv RE$.
Combining all the steps yields the desired statement.
\end{proof}

\section{Generalised entropy accumulation}  \label{sec:eat}
We are finally ready to state and prove the main result of this work which is a generalisation of the EAT proven in~\cite{eat}.
We first state a simple version of this theorem, which follows readily from the chain rule \cref{thm:improved_entropy_chain_rule} and captures the essential feature of entropy accumulation: the min-entropy of a state $\cM_n \circ \dots \circ \cM_1(\rho)$ produced by applying a sequence of $n$ channels can be lower-bounded by a sum of entropy contributions of each channel $\cM_i$.
However, for practical applications, it is desirable not to consider the state $\cM_n \circ \dots \circ \cM_1(\rho)$, but rather that state conditioned on some classical event, for example ``success'' in a key distribution protocol -- a concept called ``testing''.
Analogously to~\cite{eat}, we present an EAT adapted to that setting in~\cref{sec:testing}.

\subsection{Generalised EAT}
\begin{theorem}[Generalised EAT] \label{thm:improved_eat}
Consider a sequence of channels $\cM_i \in \cptp(R_{i-1} E_{i-1}, A_i R_i E_i)$ such that for all $i \in \{1, \dots, n\}$, there exists $\cR_i \in \cptp(E_{i-1} , E_i )$ such that $\setft{Tr}_{A_i R_i} \circ \cM_i = \cR_i \circ \setft{Tr}_{R_{i-1}}$. Then for any $\eps \in (0,1)$ and any $\rho_{R_0 E_0 } \in \states(R_0 E_0 )$
\begin{align*}
\hmin^\eps(A^n | E_n )_{\cM_n \circ \dots \circ \cM_1(\rho_{R_0 E_0 })} \geq \sum_{i = 1}^n \inf_{\omega \in \states(R_{i-1}E_{i-1}\tilde E_{i-1})} H(A_i|E_i \tilde E_{i-1} )_{\cM_i(\omega)} - O(\sqrt{n})
\end{align*}
for a purifying system $\tilde E_{i-1} \equiv R_{i-1}E_{i-1}$. For a statement with explicit constants, see \cref{eqn:eat_info_explicit} in the proof.
\end{theorem}

\begin{proof}
By \cite[Lemma B.10]{eat}, we have for $\alpha \in (1,2)$
\begin{align*}
\hmin^{\eps}(A_1^n | E_n )_{\cM_n \circ \dots \circ \cM_1(\rho_{R_0 E_0 })} \geq \had(A_1^n | E_n )_{\cM_n \circ \dots \circ \cM_1(\rho_{R_0 E_0 })} - \frac{g(\eps)}{\alpha-1}
\end{align*}
with $g(\eps) = \log(1 - \sqrt{1-\eps^2})$.
From \cref{thm:improved_entropy_chain_rule}, we have
\begin{align*}
&\had(A_1^n | E_n )_{\cM_n \circ \dots \circ \cM_1(\rho_{R_0 E_0 })} \\
&\hspace{20mm}\geq \had(A_1^{n-1} | E_{n-1} )_{\cM_{n-1} \circ \dots \circ \cM_1(\rho_{R_0 E_0 })}  + \inf_{\omega \in \states(R_{n-1}E_{n-1}\tilde E_{n-1})} \hadvar{\frac{1}{2-\alpha}}(A_n|E_n \tilde E_{n-1})_{\cM_n(\omega)} \,.
\end{align*}
Repeating this step $n-1$ times, we get 
\begin{align*}
\had(A_1^n | E_n )_{\cM_n \circ \dots \circ \cM_1(\rho_{R_0 E_0 })} 
&\geq \had(A_1 | E_1 )_{\cM_1(\rho_{R_0 E_0 })} + \sum_{i = 2}^n \inf_{\omega \in \states(R_{i-1}E_{i-1}\tilde E_{i-1})} \hadvar{\frac{1}{2-\alpha}}(A_i|E_i \tilde E_{i-1})_{\cM_i(\omega)} \\
&\geq \sum_{i = 1}^n \inf_{\omega \in \states(R_{i-1}E_{i-1}\tilde E_{i-1})} \hadvar{\frac{1}{2-\alpha}}(A_i|E_i \tilde E_{i-1})_{\cM_i(\omega)} \,,
\end{align*}
where the final step uses the monotonicity of the R\'enyi divergence in $\alpha$~\cite[Corollary~4.3]{tomamichel2015quantum}.
From \cite[Lemma B.9]{eat} we have for each $i \in \{1, \dots, n\}$ and $\alpha$ sufficiently close to 1, 
\begin{multline*}
\inf_{\omega \in \states(R_{i-1}E_{i-1}\tilde E_{i-1})} \hadvar{\frac{1}{2-\alpha}}(A_i|E_i \tilde E_{i-1})_{\cM_i(\omega)} \\
\geq  \inf_{\omega \in \states(R_{i-1}E_{i-1}\tilde E_{i-1})} H(A_i|E_i \tilde E_{i-1})_{\cM_i(\omega)} - \frac{\alpha - 1}{2-\alpha} \log^2\big(1 + 2 \, \dim(A_i)\big) \,.
\end{multline*}
Setting $d_A = \max_i \dim(A_i)$ and combining the previous steps, we obtain 
\begin{multline}
\hmin(A_1^n | E_n )_{\cM_n \circ \dots \circ \cM_1(\rho_{R_0 E_0})} \\
\geq \sum_{i = 1}^n \inf_{\omega_i \in \states(R_{i-1}E_{i-1}\tilde E_{i-1})} H(A_i|E_i \tilde E_{i-1})_{\cM_i(\omega_i)} - n \, \frac{\alpha - 1}{2-\alpha} \log^2(1 + 2 d_A) - \frac{g(\eps)}{\alpha-1} \,. \label{eqn:eat_info_explicit}
\end{multline}
Using $\alpha = 1 + O(1/\sqrt{n})$ yields the result.
\end{proof}

\subsection{Generalised EAT with testing} \label{sec:testing}
In this section, we will extend \cref{thm:improved_eat} to include the possibility of ``testing'', i.e., of computing the min-entropy of a cq-state conditioned on some classical event.
This analysis is almost identical to that of~\cite{dupuis2019entropy}; we give the full proof for completeness, but will appeal to~\cite{dupuis2019entropy} for specific tight bounds.
The resulting EAT (\cref{thm:with_testing}) has (almost) the same tight bounds as the result in~\cite{dupuis2019entropy}, but replaces the Markov condition with the more general non-signalling condition.
Hence, relaxing the Markov condition does not result in a significant loss in parameters (including second-order terms).

Consider a sequence of channels $\cM_i \in \cptp(R_{i-1} E_{i-1}, C_i A_i R_i E_i)$ for $i \in \{1, \dots, n\}$, where $C_i$ are classical systems with common alphabet $\cC$. 
We require that these channels $\cM_i$ satisfy the following condition: defining $\cM'_i  = \setft{Tr}_{C_i} \circ \cM_i$, there exists a channel $\cT \in \cptp(A^n E_n, C^n A^n E_n)$ such that $\cM_n \circ \dots \circ \cM_1 = \cT \circ \cM'_n \circ \dots \circ \cM'_1$ and $\cT$ has the form
\begin{align}
\cT(\omega_{A^n E_n}) = \sum_{y \in \cY , z \in \cZ} (\Pi_{A^n}^{(y)} \otimes \Pi_{E_n}^{(z)}) \omega_{A^n E_n} (\Pi_{A^n}^{(y)} \otimes \Pi_{E_n}^{(z)}) \otimes \proj{r(y,z)}_{C^n} \,, \label{eqn:measurement_condition}
\end{align}
where $\{\Pi_{A^n}^{(y)}\}_y$ and $\{\Pi_{E_n}^{(z)}\}_z$ are families of mutually orthogonal projectors on $A^n$ and $E_n$, and  $r : \cY \times \cZ  \to \cC$ is a deterministic function.  
Intuitively, this condition says that the classical statistics can be reconstructed ``in a projective way'' from systems $A^n$ and $E_n$ at the end of the protocol.
In particular, this requirement is always satisfied if the statistics are computed from classical information contained in $A^n$ and $E_n$, which is the case for many applications.
We note that the statistics are still generated in a round-by-round manner; \cref{eqn:measurement_condition} merely asserts that they could be \textit{reconstructed} from the final state.

Let $\mbP$ be the set of probability distributions on the alphabet $\cC$ of $C_i$, and let $\tilde E_{i-1}$ be a system isomorphic to $R_{i-1} E_{i-1}$.
For any $q \in \mbP$ we define the set of states 
\begin{align}
\label{eqn:def_sigma}
  \Sigma_i(q) = \bigl\{\nu_{C_i A_i R_i E_i \tilde E_{i-1}} = \cM_i(\omega_{R_{i-1} E_{i-1} \tilde E_{i-1}}) \,|\,  \omega \in \states(R_{i-1}E_{i-1}\tilde E_{i-1}) \text{ and } \nu_{C_i} = q \bigr\}  \ ,
\end{align}
where $\nu_{C_i}$ denotes the probability distribution over $\cC$ with the probabilities given by $\pr{c} = \bra{c} \nu_{C_i} \ket{c}$. In other words, $\Sigma_i(q)$ is the set of states that can be produced at the output of the channel $\mathcal{M}_i$ and whose reduced state on $C_i$ is equal to the probability distribution $q$.
\begin{definition} \label{def:tradeoff}
A  function $f: \mbP \to \R$ is called a \emph{min-tradeoff function} for $\{\cM_i\}$ if it satisfies 
\begin{align*}
f(q) \leq \min_{\nu \in \Sigma_i(q)} H(A_i|E_i \tilde E_{i-1})_{\nu} \quad \forall i = 1, \dots, n\, .
\end{align*}
Note that if $\Sigma_i(q) = \emptyset$, then $f(q)$ can be chosen arbitrarily.
\end{definition}
Our result will depend on some simple properties of the tradeoff function, namely the maximum and minimum of $f$, the minimum of $f$ over valid distributions, and the maximum variance of $f$:
\begin{align*}
\Max{f} &\deq \max_{q \in \mathbb{P}} f(q) \,,\\
\Min{f} &\deq \min_{q \in \mathbb{P}} f(q) \,,\\
\MinSigma{f} &\deq \min_{q : \Sigma(q) \neq \emptyset} f(q) \,,\\
\Var{f} &\deq \max_{ q : \Sigma(q) \neq \emptyset} \sum_{x \in \cC} q(x) f(\delta_{x})^2 - \left(\sum_{x \in \cC} q(x) f(\delta_x) \right)^2 \,,
\end{align*}
where $\Sigma(q) = \bigcup_i \Sigma_i(q)$ and $\delta_x$ is the distribution with all the weight on element $x$.
We write $\freq{C^n}$ for the distribution on $\cC$ defined by $\freq{C^n}(c) = \frac{|\{i \in \{1,\dots,n\} : C_i = c\}|}{n}$.  
We also recall that in this context, an event $\Omega$ is defined by a subset of $\cC^n$, and for a state $\rho_{C^n A^n E_n R_n}$ we write $\prs{\rho}{\Omega} = \sum_{c^n \in \Omega}\tr{\rho_{A_1^n E_n R_n, c^n}}$ for the probability of the event $\Omega$ and
\begin{align*}
  \rho_{C^n A^n E_n R_n | \Omega} = \frac{1}{\prs{\rho}{\Omega}} \sum_{c^n \in \Omega} \proj{c^n}_{C^n} \otimes \rho_{A^n E_n R_n, c^n}
\end{align*}
for the state conditioned on~$\Omega$.

\begin{theorem}\label{thm:with_testing}
Consider a sequence of channels $\cM_i \in \cptp(R_{i-1}E_{i-1}, C_i A_i R_i E_i)$ for $i \in \{1, \dots, n\}$, where $C_i$ are classical systems with common alphabet $\cC$ and the sequence $\{\cM_i\}$ satisfies \cref{eqn:measurement_condition} and the non-signalling condition: for each $\cM_i$, there exists $\cR_i \in \cptp(E_{i-1}, E_i)$ such that $\setft{Tr}_{A_i R_i C_i} \circ \cM_i = \cR_i \circ \setft{Tr}_{R_{i-1}}$.
Let $\eps \in (0,1)$, $\alpha \in (1, 3/2)$, $\Omega \subset \cC^n$, $\rho_{R_0 E_0} \in \states(R_0 E_0)$, and $f$ be an affine\footnote{A function $f$ on the convex set $\mbP(\cC)$ is called \emph{affine} if it is linear under convex combinations, i.e., for $\lambda \in [0,1]$ and $p_1, p_2 \in \mbP(\cC)$, $\lambda f(p_1) + (1 - \lambda) f(p_2) = f(\lambda p_1 + (1 - \lambda) p_2)$. Such functions are also sometimes called \emph{convex-linear}.} min-tradeoff function with $h = \min_{c^n \in \Omega} f(\freq{c^n})$. Then,
\begin{multline}
\hmin^\eps(A^n | E_n)_{\cM_n \circ \dots \circ \cM_1(\rho_{R_0 E_0})_{|\Omega}} \\
\geq n \, h - n \, \frac{\alpha-1}{2-\alpha} \, \frac{\ln(2)}{2} V^2 - \frac{g(\eps) + \alpha \log(1/\prs{\rho^n}{\Omega})}{\alpha-1} -  n \, \left( \frac{\alpha-1}{2-\alpha} \right)^2 K'(\alpha)\, ,  \label{eqn:alpha_to_choose}
\end{multline}
where $\pr{\Omega}$ is the probability of observing event $\Omega$, and
\begin{align*}
g(\eps) &= - \log(1 - \sqrt{1-\eps^2}) \,,\\
V &= \log (2d_A^2+1) + \sqrt{2 + \Var{f}} \,,\\
K'(\alpha) &= \frac{(2-\alpha)^3}{6 (3-2\,\alpha)^3 \ln 2} \, 2^{\frac{\alpha-1}{2-\alpha}(2\log d_{A}  + \Max{f} - \MinSigma{f})} \ln^3\left( 2^{2\log d_{A} + \Max{f} - \MinSigma{f}} + e^2 \right) \,,
\end{align*}
with $d_A = \max_i \dim(A_i)$.
\end{theorem}

\begin{remark}
The parameter in~$\alpha$ in~\cref{thm:with_testing} can be optimized for specific problems, which leads to tighter bounds. Alternatively, it is possible to make a generic choice for $\alpha$ to recover a theorem that looks much more like \cref{thm:improved_eat}, which is done in \cref{cor:alpha_chosen}.
We also remark that even tighter second order terms have been derived in \cite{liu2021device}.
To keep our theorem statement and proofs simpler, we do not carry out this additional optimization explicitly, but note that this can be done in complete analogy to \cite{liu2021device}.
\end{remark}

To prove \cref{thm:with_testing}, we will need the following lemma (which is already implicit in \cite[Claim 4.6]{eat}, but we give a simplified proof here).
\begin{lemma} \label{lem:chain_rule_testing}
Consider a quantum state $\rho \in \states(CADE)$ that has the form 
\begin{align*}
\rho_{CADE} = \sum_{c \in \Omega} \proj{c} \otimes \rho_{AE,c} \otimes \rho_{D|c} \,,
\end{align*}
where $\Omega \subset \cC$ is a subset of the alphabet $\cC$ of the classical system $C$, and for each $c$, $\rho_{AE,c} \in \pos(AE)$ is subnormalised and $\rho_{D|c} \in \states(D)$ is a quantum state.
Then for $\alpha > 1$,
\begin{align*}
\hau(ACD|E)_{\rho} \leq \hau(AC|E)_{\rho} + \max_{c \in \Omega} \had(D)_{\rho_{D|c}} \,.
\end{align*}
\end{lemma}
\begin{proof}
Let $\sigma_E \in \states(E)$ such that 
\begin{align*}
\hau(ACD|E)_{\rho} = - \dalpha{\rho_{CADE}}{\1_{CAD} \ot \sigma_E} \,.
\end{align*}
Then
\begin{align*}
\left( \sigma_E^{\frac{1-\alpha}{2 \alpha}} \rho_{CADE} \sigma_E^{\frac{1-\alpha}{2 \alpha}} \right)^{\alpha} 
= \sum_{c \in \Omega} \proj{c} \ot \left( \sigma_E^{\frac{1-\alpha}{2 \alpha}} \rho_{AE,c} \sigma_E^{\frac{1-\alpha}{2 \alpha}} \right)^{\alpha}  \ot \rho_{D|c}^{\alpha} \,.
\end{align*}
Hence, 
\begin{align*}
\tr{\left( \sigma_E^{\frac{1-\alpha}{2 \alpha}} \rho_{CADE} \sigma_E^{\frac{1-\alpha}{2 \alpha}} \right)^{\alpha}} 
&= \sum_{c \in \Omega} \tr{\left( \sigma_E^{\frac{1-\alpha}{2 \alpha}} \rho_{AE,c} \sigma_E^{\frac{1-\alpha}{2 \alpha}} \right)^{\alpha}}  \, \tr{\rho_{D|c}^{\alpha}} \\
&\leq \sup_{\tilde \sigma_E \in \states(E)} \tr{\sum_{c \in \Omega} \proj{c} \ot \left( \tilde \sigma_E^{\frac{1-\alpha}{2 \alpha}} \rho_{AE,c} \tilde \sigma_E^{\frac{1-\alpha}{2 \alpha}} \right)^{\alpha}}  \, \max_{c \in \Omega} \tr{\rho_{D|c}^{\alpha}} \\
&= \sup_{\tilde \sigma_E \in \states(E)} \tr{\left( \tilde \sigma_E^{\frac{1-\alpha}{2 \alpha}} \rho_{CAE} \tilde \sigma_E^{\frac{1-\alpha}{2 \alpha}} \right)^{\alpha}} \, \max_{c \in \Omega} \tr{\rho_{D|c}^{\alpha}}
\end{align*}
Recalling the definitions of $D_{\alpha}$ (\cref{def:dalpha}) and $\hau$ (\cref{def:halpha}), we see that the lemma follows by taking the logarithm and multiplying by $\frac{1}{\alpha - 1}$.
\end{proof}

\begin{proof}[Proof of \cref{thm:with_testing}]
As in the proof of \cref{thm:improved_eat}, we first use \cite[Lemma B.10]{eat} to get
\begin{align}
\hmin^\eps(A^n | E_n)_{\cM_n \circ \dots \circ \cM_1(\rho_{R_0 E_0 })_{|\Omega}} \geq \hau(A^n | E_n )_{\cM_n \circ \dots \circ \cM_1(\rho_{R_0 E_0 })_{|\Omega}} - \frac{g(\eps)}{\alpha-1} \label{eqn:min_to_alpha}
\end{align}
for $\alpha \in (1,2]$ and $g(\eps) = \log(1 - \sqrt{1-\eps^2})$.
We therefore need to find a lower bound for 
\begin{align}
\hau(A^n | E_n )_{\cM_n \circ \dots \circ \cM_1(\rho_{R_0 E_0 })_{|\Omega}} = \hau(A^n C^n| E_n )_{\cM_n \circ \dots \circ \cM_1(\rho_{R_0 E_0 })_{|\Omega}} \,,\label{eqn:halpha_cond}
\end{align}
where the equality holds because of \cref{eqn:measurement_condition} and \cite[Lemma B.7]{eat}.

Before proceeding with the formal proof, let us explain the main difficulty compared to \cref{thm:improved_eat}.
The state for which we need to compute the entropy in \cref{eqn:halpha_cond} is conditioned on the event $\Omega \subset \cC^n$. 
This is a global event, in the sense that it depends on the classical outputs $C_1, \dots, C_n$ of all rounds.
We essentially seek a lower bound that involves $\min_{\nu \in \Sigma_i(\freq{c^n})} \had(A_i|E_i )_{\nu}$ for some $c^n \in \Omega$, i.e., for every round we only want to minimize over output states of the channel $\cM_i$ whose distribution on $C_i$ matches the frequency distribution $\freq{c^n}$ of the $n$ rounds we observed.
This means that we must use the global conditioning on $\Omega$ to argue that in each round, we can restrict our attention to states whose outcome distribution matches the (worst-case) frequency distribution associated with $\Omega$.
The chain rule \cref{thm:improved_chain_rule} does not directly allow us to do this as the r.h.s.~of \cref{eqn:entropy_chain_rule} always minimizes over all possible input states.

To circumvent this, we follow a strategy that was introduced in \cite{eat} and optimized in \cite{dupuis2019entropy}.
For every $i$, we introduce a quantum system $D_i$ with $\dim(D_i) = \lceil  2^{\Max{f} - \Min{f}} \rceil$ and define $\cD_i \in \cptp(C_i, C_i D_i)$ by 
\begin{align*}
\cD_i(\omega_{C_i}) = \sum_{c \in \cC} \bra{c}\omega_{C_i}\ket{c} \cdot \proj{c} \ot \tau_{D_{i}|c} \,.
\end{align*}
For every $c \in \cC$, the state $\tau_{D_{i}|c} \in \states(D)$ is defined as the mixture between a uniform distribution on $\{1, \dots, \lfloor 2^{\Max{f} - f(\delta_c)} \rfloor \}$ and a uniform distribution on $\{1, \dots, \lceil 2^{\Max{f} - f(\delta_c)} \rceil\}$ that satisfies 
\begin{align*}
H(D_i)_{\tau_{D_i|c}} = \Max{f} - f(\delta_c) \, ,
\end{align*} 
where $\delta_x$ stands for the distribution with all the weight on element $x$. 
This is clearly possible if $\dim(D_i) = \lceil  2^{\Max{f} - \Min{f}} \rceil$. 

We define $\bar \cM_i = \cD_i \circ \cM_i$ and denote 
\begin{align*}
\rho^n_{C^n A^n R_n E_n } = \cM_n \circ \dots \circ \cM_1(\rho_{R_0 E_0 }) \;\tand\; \bar \rho^n_{C^n A^n D^n R_n E_n } = \bar \cM_n \circ \dots \circ \bar \cM_1(\rho_{R_0 E_0 }) \,. 
\end{align*} 
The state $\bar \rho^n_{|\Omega}$ has the right form for us to apply \cref{lem:chain_rule_testing} and get 
\begin{align}
\hau(A^n C^n |E_n )_{\bar \rho^n_{|\Omega}} \geq - \max_{c^n \in \Omega} \had(D^n)_{\bar \rho^n_{D^n|c^n}} + \hau(A^n C^n D^n|E_n )_{\bar \rho^n_{|\Omega}}  \,, \label{eqn:h_testing_overall}
\end{align}
where 
\begin{align*}
\bar \rho^n_{D^n|c^n} = \tau_{D_{1}|c_1} \ot \dots \ot \tau_{D_{n}|c_n} \,.
\end{align*}
We treat each term in \cref{eqn:h_testing_overall} in turn.
\begin{enumerate}
\item For the term on the l.h.s., it is easy to see that $\bar \rho^n_{C^n A^n R_n E_n |\Omega} = \rho^n_{C^n A^n R_n E_n |\Omega}$, so 
\begin{align}
\hau(A^n C^n |E_n )_{\bar \rho^n_{|\Omega}} = \hau(A^n C^n |E_n )_{\rho^n_{|\Omega}} \,. \label{eqn:h_testing1} 
\end{align}
\item For the first term on the r.h.s., we compute
\begin{align*}
\had(D^n)_{\bar \rho^n_{D^n|c^n}} 
= \sum_{i} \had(D_i)_{\tau_{D_{i}|c_i}} 
\leq \sum_{i} H(D_i)_{\tau_{D_{i}|c_i}} 
&= n\,\Max{f} - \sum_{i} f(\delta_{c_i}) \\
&= n\,\Max{f} - n f(\freq{c^n}) \,, \numberthis \label{eqn:h_testing_2}
\end{align*} 
where the last equality holds because $f$ is affine.
\item For the second term on the r.h.s., we first use \cite[Lemma B.5]{eat} to remove the conditioning on the event $\Omega$, and then use that removing the classical system $C^n$ and switching from $\hau$ to $\had$ can only decrease the entropy: 
\begin{align*}
\hau(A^n C^n D^n|E_n )_{\bar \rho^n_{|\Omega}} 
&\geq \had(A^n D^n|E_n )_{\bar \rho^n} - \frac{\alpha}{\alpha-1} \log(1/\prs{\rho^n}{\Omega}) \,,
\end{align*}
where we used $\prs{\rho^n}{\Omega} = \prs{\bar \rho^n}{\Omega}$. 
Now noting that $\setft{Tr}_{D_i} \circ \bar \cM_i = \cM_i$, we see that the non-signalling condition $\setft{Tr}_{A_i R_i C_i} \circ \cM_i = \cR_i \circ \setft{Tr}_{R_{i-1}}$ on $\cM_i$ implies the non-signalling condition $\setft{Tr}_{A_i R_i C_i D_i} \circ \bar \cM_i = \cR_i \circ \setft{Tr}_{R_{i-1}}$ on $\bar \cM_i$.
We can therefore apply the chain rule in \cref{thm:improved_entropy_chain_rule} to find 
\begin{align*}
\had(A^n D^n|E_n )_{\bar \rho^n} \geq \sum_{i = 1}^n \min_{\omega_{i-1} \in \states(R_{i-1} E_{i-1} \tilde E_{i-1})} \hadvar{\beta}(A_i D_i|E_i \tilde E_{i-1})_{\bar \cM_i(\omega_{i-1})} \,,
\end{align*}
where we introduced the shorthand $\beta \deq \frac{1}{2-\alpha}$ and the purifying system $\tilde E_{i-1} \equiv R_{i-1} E_{i-1}$.
Noting that for $\alpha \in (1,3/2)$ we have $\beta \in (1, 2)$, we can now use \cite[Corollary IV.2]{dupuis2019entropy} to obtain 
\begin{align*}
\hadvar{\beta}(A_i D_i|E_i \tilde E_{i-1})_{\bar \cM_i(\omega_{i-1})} \geq H(A_i D_i|E_i \tilde E_{i-1})_{\bar \cM_i(\omega_{i-1})} - (\beta - 1)\frac{\ln(2)}{2} V^2 - (\beta-1)^2 K( \beta ) \,,
\end{align*}
where $V^2$ and $K(\beta)$ are quantities from~\cite[Proposition V.3]{dupuis2019entropy} that satisfy 
\begin{align*}
K(\beta) &\leq \frac{1}{6 (2-\beta)^3 \ln 2} \, 2^{(\beta - 1)(2\log d_{A}  + \Max{f} - \MinSigma{f})} \ln^3\left( 2^{2\log d_{A} + \Max{f} - \MinSigma{f}} + e^2 \right) \,, \\
 V^2 &= \left(\log (2d_A^2+1) + \sqrt{2 + \Var{f}} \right)^2 \,,
\end{align*}
where $d_A = \max_i \dim(A_i)$.
Furthermore, as in the proof of \cite[Proposition V.3]{dupuis2019entropy}, we have 
\begin{align*}
H(A_i D_i|E_i \tilde E_{i-1})_{\bar \cM_i(\omega_{i-1})} \geq \Max{f} \,.
\end{align*}
Therefore, the second term on the r.h.s.~of \cref{eqn:h_testing_overall} is bounded by
\begin{align}
&\hau(A^n C^n D^n|E_n)_{\bar \rho^n_{|\Omega}} \nonumber \\
&\hspace{10mm}\geq n \, \Max{f}  - n \, (\beta - 1)\frac{\ln(2)}{2} V^2 - n \, (\beta - 1)^2 K(\beta) - \frac{\alpha}{\alpha-1} \log(1/\prs{\rho^n}{\Omega})\,. \label{eqn:h_testing3}
\end{align}
\end{enumerate}
Combining our results for each of the three terms (i.e.~\cref{eqn:h_testing1}, \cref{eqn:h_testing_2}, and \cref{eqn:h_testing3}) and recalling $h = \min_{x^n \in \Omega} f(\freq{x^n})$, \cref{eqn:h_testing_overall} becomes 
\begin{align*}
\hau(A^n C^n |E_n )_{\rho^n_{|\Omega}} \geq n \, h - n \, (\beta - 1)\frac{\ln(2)}{2} V^2 - \frac{\alpha}{\alpha-1} \log(1/\prs{\rho^n}{\Omega}) - n \, (\beta - 1)^2 K(\beta) \,.
\end{align*}
Inserting this into \cref{eqn:min_to_alpha} and \cref{eqn:halpha_cond}, and defining $K'(\alpha) = K(\beta) = K(\frac{1}{2-\alpha})$ we obtain
\begin{multline}
\hmin^\eps(A^n | E_n )_{\cM_n \circ \dots \circ \cM_1(\rho_{R_0 E_0 })_{|\Omega}} \\
\geq n \, h - n \, (\beta - 1)\frac{\ln(2)}{2} V^2 - \frac{g(\eps) + \alpha \log(1/\prs{\rho^n}{\Omega})}{\alpha-1} -  n \, (\beta - 1)^2 K(\beta) 
\end{multline}
as desired.
\end{proof}

\begin{corollary} \label{cor:alpha_chosen}
For the setting given in~\cref{thm:with_testing} we have
\begin{align*}
\hmin^\eps(A^n | E_n )_{\cM_n \circ \dots \circ \cM_1(\rho_{R_0 E_0 })_{|\Omega}} \geq n h - c_1 \sqrt{n} - c_0 \, ,
\end{align*}
where the constants $c_1$ and $c_0$ are given by
\begin{align*}
c_1 &= \sqrt{ \frac{2 \ln(2) V^2}{\eta}\Big(g(\eps) + (2 - \eta) \log(1/\prs{\rho^n}{\Omega}) \Big)}\,,\\
c_0 &= \frac{(2 - \eta) \eta^2 \log(1/\prs{\rho^n}{\Omega}) + \eta^2 g(\eps)}{3 (\ln 2)^2 V^2 (2\eta-1)^3} \, 2^{\frac{1-\eta}{\eta}(2\log d_{A}  + \Max{f} - \MinSigma{f})} \ln^3\left( 2^{2\log d_{A} + \Max{f} - \MinSigma{f}} + e^2 \right)
\end{align*}
with 
\begin{align*}
\eta &= \frac{2 \ln(2)}{1 + 2 \ln(2)}  
\,, \qquad g(\eps) = \log(1 - \sqrt{1-\eps^2}) \,, \qquad
V = \log (2d_A^2+1) + \sqrt{2 + \Var{f}} \,.
\end{align*}
\end{corollary}
\begin{proof}
We first note that for any $\Omega$ with non-zero probability, $h \leq \log d_A$.
Therefore, if $n \leq \left( \frac{c_1}{2 \log d_A} \right)^2$, it is easy to check that $n h - c_1 \sqrt{n} \leq - n \log d_A$, so the statement of \cref{cor:alpha_chosen} becomes trivial.
We may therefore assume that $n \geq ( \frac{c_1}{2 \log d_A})^2$.

As in the proof of \cref{thm:with_testing}, we define $\beta = \frac{1}{2-\alpha}$.
We will assume that $\alpha \in (1, 2-\eta)$ for $\eta = \frac{2 \ln(2)}{1+ 2 \ln(2)} \approx 0.58$ and later make a choice of $\alpha$ that satisfies this condition.
Then, $\beta - 1 = \frac{1}{2 - \alpha} - 1 \leq \frac{\alpha - 1}{\eta}$ and $\beta \in (1, 1/\eta)$.
Therefore, using $K(\beta)$ as defined in the proof of \cref{thm:with_testing} and noting that in the interval $\beta \in (1, 1/\eta) \subset (1,2)$ this quantity is monotonically increasing in $\beta$, we have 
\begin{align*}
K(\beta) &\leq K \deq \frac{\eta^3}{6 (2 \eta - 1)^3 \ln 2} \, 2^{\frac{1-\eta}{\eta}(2\log d_{A}  + \Max{f} - \MinSigma{f})} \ln^3\left( 2^{2\log d_{A} + \Max{f} - \MinSigma{f}} + e^2 \right) \,,
\end{align*}
Hence, we can simplify the statement of \cref{thm:with_testing} to 
\begin{multline}
\hmin^\eps(A^n | E_n )_{\cM_n \circ \dots \circ \cM_1(\rho_{R_0 E_0 })_{|\Omega}} \\
\geq n \, h - n \, (\alpha - 1) \frac{\ln(2)}{2 \eta} V^2 - \frac{g(\eps) + (2-\eta) \cdot \log(1/\prs{\rho^n}{\Omega})}{\alpha-1} -  n \, (\alpha - 1)^2 \frac{K}{\eta^2}\,. \label{eqn:bound_inter}
\end{multline}
We now choose $\alpha > 1$ as a function of $n$ and $\eps$ so that the terms proportional to $\alpha - 1$ and $\frac{1}{\alpha-1}$ match: 
\begin{align*}
\alpha = 1+ \sqrt{\frac{2 \eta}{n  \ln(2) V^2 }\Big( g(\eps) + (2 - \eta) \log(1/\prs{\rho^n}{\Omega}) \Big)} \,. 
\end{align*}
Inserting this choice of $\alpha$ into \cref{eqn:bound_inter} and combining terms yields the constants in \cref{cor:alpha_chosen}.
The final step is to show that this choice of $\alpha$ indeed satisfies $\alpha \leq 2-\eta$ for $n \geq( \frac{c_1}{2 \log d_A})^2$. For this, we note that for $n \geq( \frac{c_1}{2 \log d_A})^2$
\begin{align*}
\alpha = 1 + \frac{\eta}{\ln(2) V^2} \frac{c_1}{\sqrt{n}} \leq 1 + \frac{2 \eta \log d_A}{\ln(2) V^2} \,.
\end{align*}
We can now use that $V^2 \geq \left( \log(2 d_A^2) \right)^2 \geq 4 \log d_A$ since $d_A \geq 2$, so 
\begin{align*}
\alpha \leq 1 + \frac{2 \eta \log d_A}{\ln(2) V^2} \leq  1 + \frac{\eta}{2 \ln(2)} = 2 - \eta \,,
\end{align*}
where the last inequality holds because $\eta = \frac{2 \ln(2)}{1 + 2 \ln(2)}$. 
\end{proof}

In many applications,  e.g.~randomness expansion or QKD, a round can either be a ``data generation round'' (e.g.~to generate bits of randomness or key) or a ``test round'' (e.g.~to test whether a device used in the protocol behaves as intended).
More formally, in this case the maps $\cM_i \in \cptp(R_{i-1}E_{i-1}, C_i A_i R_i E_i)$ can be written as 
\begin{align}
\cM_i = \gamma \cM^{\textnormal{test}}_{i, R_{i-1}E_{i-1} \to C_i A_i R_i E_i} + (1-\gamma) \cM^{\textnormal{data}}_{i, R_{i-1}E_{i-1} \to A_i R_i E_i} \ot \proj{\bot}_{C_i} \,, \label{eqn:testing_decomp}
\end{align}
where the output of $\cM^{\textnormal{test}}_i$ on system $C_i$ is from some alphabet $\cC'$ that does not include $\bot$, so the alphabet of system $C_i$ is $\cC = \cC' \cup \{\bot\}$.
The parameter $\gamma$ is called the \emph{testing probability}, and for efficient protocols we usually want $\gamma$ to be as small as possible.

For maps of the form in \cref{eqn:testing_decomp}, there is a general way of constructing a min-tradeoff function for the map $\cM_i$ based only on the statistics generated by the map $\cM^{\textnormal{test}}_i$.
This was shown in \cite{dupuis2019entropy} and we reproduce their result (adapted to our notation) here for the reader's convenience.

\begin{lemma}[{\cite[Lemma V.5]{dupuis2019entropy}}]
\label{lem:infrequent_sampling}
Let $\mathcal{M}_i \in \cptp(R_{i-1}E_{i-1}, C_i A_i R_i E_i)$ be channels satisfying the same conditions as in \cref{thm:with_testing} that can furthermore be decomposed as in \cref{eqn:testing_decomp}. 
Suppose that an affine function $g : \mbP(\cC') \to \R$ satisfies for any $q' \in \mbP(\cC')$ and any $i = 1, \dots, n$
\begin{align}
\label{eq:tradeoff_g}
g(q') \leq \min_{\omega \in \states(R_{i-1} E_{i-i} \tilde E_{i-1})} \bigl\{ H(A_i | E_i \tilde E_{i-1})_{\cM_i(\omega)}:  \left(\cM^{\textnormal{test}}_i(\omega)\right)_{C_i} = q' \bigr\} 
\end{align}
where $\tilde E_{i-1} \equiv R_{i-1} E_{i-1}$ is a purifying system.
Then, the affine function $f : \mbP(\cC) \to \R$ defined by
\begin{align*}
f(\delta_x) &=  \Max{g} + \frac{1}{\gamma} (g(\delta_x) - \Max{g}) \;\;\;\;\; \forall x \in \mathcal{C}'\\
f(\delta_{\bot}) &= \Max{g} 
\end{align*}   
is a min-tradeoff function for $\{\mathcal{M}_{i}\}$. Moreover,
\begin{align*}
\Max{f} &= \Max{g} \\
\Min{f} &= \left(1-\frac{1}{\gamma} \right) \Max{g} + \frac{1}{\gamma}\Min{g} \\
\MinSigma{f} &\geq \Min{g} \\
 \Var{f} &\leq \frac{1}{\gamma}\big(\Max{g} - \Min{g} \big)^2.
 \end{align*}
\end{lemma}

\section{Sample applications}\label{sec:application_eat}
To demonstrate the utility of our generalised EAT, we provide two sample applications.
Firstly, in \cref{sec:blind_re} we prove security of blind randomness expansion against general attacks.
The notion of blind randomness was defined in~\cite{miller2017randomness} and has potential applications in mistrustful cryptography (see~\cite{miller2017randomness, fu2018local} for a detailed motivation).
Until now, no security proof against general attacks was known.
In particular, the original EAT is not applicable because its model of side information is too restrictive.
With our generalised EAT, we can show that security against general attacks follows straightforwardly from a single-round security statement.

Secondly, in \cref{sec:qkd} we give a simplified security proof for the E91 QKD protocol~\cite{E91}, which was also treated with the original EAT~\cite{eat}.
This example is meant to help those familiar with the original EAT understand the difference between that result and our generalised EAT.
In particular, this application highlights the utility of our more general model of side information: in our proof, the non-signalling condition is satisfied trivially and the advantage over the original EAT stems purely from being able to update the side information register~$E_i$.
We point out that while here we focus on the E91 protocol to allow an easy comparison with the original EAT, our generalised EAT can be used for a large class of QKD protocols for which the original EAT was not applicable at all.
A comprehensive treatment of this is given in~\cite{eat_for_qkd}.

\subsection{Blind randomness expansion} \label{sec:blind_re}
We start by recalling the idea of standard (non-blind) device-independent randomness expansion~\cite{colbeck2009quantum,Colbeck_2011,pironio2010random,vazirani2012certifiable,miller2016robust}.
Alice would like to generate a uniformly random bit string using devices $D_1$ and $D_2$ prepared by an adversary Eve.
To this end, in her local lab (which Eve cannot access) she isolates the devices from one another and plays multiple round of a non-local game with them, e.g.~the CHSH game.
On a subset of the rounds of the game, she checks whether the CHSH condition is satisfied.
If this is the case on a sufficiently high proportion of rounds, she can conclude that the devices' outputs on the remaining rounds must contain a certain amount of entropy, conditioned on the input to the devices and any quantum side information that Eve might have kept from preparing the devices.
Using a quantum-proof randomness extractor, Alice can then produce a uniformly random string.

Blind randomness expansion~\cite{miller2017randomness, fu2018local} is a significant strengthening of the above idea.
Here, Alice only receives one device $D_1$, which she again places in her local lab isolated from the outside world.
Now, Alice plays a non-local game with her device $D_1$ and the adversary Eve:
she samples questions for a non-local game as before, inputs one of the questions to $D_1$, and sends the other question to Eve.
$D_1$ and Eve both provide an output.
Alice then proceeds as in standard randomness expansion, checking whether the winning condition of the non-local game is satisfied on a subset of rounds and concluding that the output of her device $D_1$ must contain a certain amount of entropy conditioned on the adversary's side information.

For the purpose of applying the EAT, the crucial difference between the two notions of randomness expansion is the following:
in standard randomness expansion, the adversary's quantum side information is not acted upon during the protocol, and additional side information (the inputs to the devices, which we also condition on) are generated independently in a round-by-round manner.
This allows a relatively straightforward application of the standard EAT~\cite{arnon2018practical}.
In contrast, in blind randomness expansion, the adversary's quantum side information gets updated in every round of the protocol and is not generated independently in a round-by-round fashion.
This does not fit in the framework of the standard EAT, which requires the side information to be generated round-by-round subject to a Markov condition.
As a result,~\cite{miller2017randomness, fu2018local} were not able to prove a general multi-round blind randomness expansion result.

In the rest of this section, we will show that our generalised EAT is capable of treating multi-round blind randomness expansion, using a protocol similar to \cite[Protocol 3.1]{arnon2019simple}.
A formal description of the protocol is given in \cref{prot:re}.

\begin{figure}[ht!]
\begin{longfbox}[breakable=false, padding=1em, margin-top=1em, margin-bottom=1em]
\begin{protocol} {\bf General blind randomness expansion protocol} \label{prot:re} \end{protocol}
\noindent\underline{Protocol arguments}  \vspace{-0.8ex}
\vspace{-0.6ex}
\begin{center}
\begin{tabularx}{\textwidth}{r c X}
$G$ & : & two-player non-local game, specified by a question set $\cX \times \cY$, a probability distribution $q$ on $\cX \times \cY$, an answer set $\cA \times \cB$, and a winning condition $\omega: \cX \times \cY\times \cA \times \cB \to \bits$ \\
$x^* \in \cX, y^* \in \cY$ & : & inputs used for generation rounds \\
$D$ &:& untrusted device capable of playing one side of $G$ repeatedly \\
$n \in \N$ &:& number of rounds \\
$\gamma \in (0,1]$ &:& expected fraction of test rounds \\
$\wexp$ &:& expected winning probability in $G$ \\
$\delta$ &:& error tolerance
\end{tabularx}
\end{center}

\vspace{0.1ex}

\noindent\underline{Protocol steps}
\vspace{1ex}

For rounds $i = 1, \dots, n$, Alice performs the following steps:
\vspace{-0.6ex}
\begin{enumerate}[label=(\arabic*)]
    \setlength{\itemsep}{0.4ex}
    \setlength{\parskip}{0.4ex}
    \setlength{\parsep}{0.4ex} 
  \item Alice chooses $T_i \in \bits$ with $\pr{T_i = 1} = \gamma$. If $T_i = 1$, Alice chooses $X_i, Y_i \in \cX \times \cY$ according to the question distribution $q$. If $T_i = 0$, Alice chooses $X_i = x^*, Y_i = y^*$. \label{step:input}
  \item Alice inputs $X_i$ into her device $D$ and sends $Y_i$ to Eve. She receives answers $A_i$ and $B_i$, respectively.
  \item If $T_i = 0$, Alice sets $C_i = \bot$. If $T_i = 1$, Alice sets $C_i = \omega(X_i, Y_i, A_i, B_i)$.
  \end{enumerate} 
  \vspace{-1ex}
 At the end of the protocol, Alice aborts if $|\{i \sth C_i = 0\}| > (1 - \wexp + \delta) \cdot \gamma n$.
 \end{longfbox}
 \end{figure}

The following proposition shows a lower bound on on the amount of randomness Alice can extract from this protocol, as specified by the min-entropy.
For this, we assume a lower-bound on the single-round von Neumann entropy.
Such a single-round bound can be found numerically using a generic method as explained after the proof~\cref{lem:blind_re_eat}.

\begin{proposition} \label{lem:blind_re_eat}
Suppose Alice executes \cref{prot:re} with a device $D$ that cannot communicate with Eve.
We denote by $R_i$ and $E'_i$ the (arbitrary) quantum systems of the device $D$ and the adversary Eve after the $i$-th round, respectively.
Eve's full side-information after the $i$-th round is  $E_i \deq T^i X^i Y^i B^i E'_i$.
A single round of the protocol can be described by a quantum channel $\cN_i \in \cptp(R_{i-1} E_{i-1}, C_i A_i R_i E_i)$.
We also define $\cN_i^{\setft{test}}$ to be the same as $\cN_i$, except that $\cN_i^{\setft{test}}$ always picks $T_i = 1$.
Let $\rho_{A^n C^n R_n E_n}$ be the state at the end of the protocol and $\Omega$ the event that Alice does not abort.

Let $g: \mbP(\bits) \to \R$ be an affine function satisfying the conditions
\begin{align}
g(p) \leq \inf_{\omega \in \states(R_{i-1}E_{i-1}\tilde E_{i-1}) : \, \cN_i^{\setft{test}}(\omega)_{C_i} = p} H(A_i|E_i \tilde E_{i-1})_{\cN_i(\omega)} \,, \qquad \Max{g} = g(\delta_1) \,, \label{eqn:g_cond}
\end{align}
where $\tilde E_{i-1} \equiv R_{i-1}E_{i-1}$ is a purifying system.
Then, for any $\eps_a, \eps_s \in (0,1)$, either $\pr{\Omega} \leq \eps_a$ or 
\begin{align*}
\hmin^{\eps_s} (A^n | E_n) \geq n h - c_1 \sqrt{n} - c_0 
\end{align*}
for $c_1, c_0 \geq 0$ independent of $n$ and
\begin{align*}
h &= \min_{p' \in \mbP(\bits): p'(0) \leq 1 - \wexp + \delta} g(p') \,,
\end{align*}
where $\wexp$ is the expected winning probability and $\delta$ the error tolerance from \cref{prot:re}.
If we treat $\eps_s, \eps_a, \dim(A_i), \delta, \Max{g},$ and $\Min{g}$ as constants, then $c_1 = O(1/\sqrt{\gamma})$ and $c_0 = O(1)$.

Furthermore, if there exists a quantum strategy that wins the game $G$ with probability $\wexp$, there is an honest behaviour of $D$ and Eve for which $\pr{\Omega} \geq 1 - \exp(-\frac{\delta^2}{1 - \wexp + \delta} \gamma n)$.
\end{proposition}

\begin{remark}
The condition on $g(p)$ in~\cref{eqn:g_cond} is formulated in terms of the entropy 
\begin{align*}
H(A_i|E_i \tilde E_{i-1})_{\cN_i(\omega)} = H(A_i|T^i X^i Y^i B^i E'_i \tilde E_{i-1})_{\cN_i(\omega)}
\end{align*}
with $\tilde E_{i-1} \equiv R_{i-1} E_{i-1}$.
However, the map $\cN_i$ corresponding to the $i$-th round does not act on the systems $T^{i-1} X^{i-1} Y^{i-1} B^{i-1}$.
Therefore, we can view these systems as part of the purifying system.
Since the infimum in~\cref{eqn:g_cond} already includes a purifying $\tilde E_{i-1}$, we can drop these additional systems and without loss of generality choose $\tilde E_{i-1}$ to be isomorphic to those input systems on which $\cN_i$ acts non-trivially, i.e.~$\tilde E_{i-1} \equiv R_{i-1} E'_{i-1}$.
This means that we can replace the upper bound on $g$ in \cref{eqn:g_cond} by the equivalent condition 
\begin{align}
g(p) \leq \inf_{\omega \in \states(R_{i-1}E_{i-1}\tilde E_{i-1}) : \, \cN_i^{\setft{test}}(\omega)_{C_i} = p} H(A_i|B_i X_i Y_i T_i E'_i \tilde E_{i-1})_{\cN_i(\omega)} \label{eqn:g_cond_mod}
\end{align}
with $\tilde E_{i-1} \equiv R_{i-1} E'_{i-1}$.
For the proof of \cref{lem:blind_re_eat} we will use \cref{eqn:g_cond} since it more closely matches the notation of \cref{thm:with_testing}, but intuitively, \cref{eqn:g_cond_mod} is more natural as it only involves quantities related to the $i$-th round of the protocol.
\end{remark}

\begin{proof}[Proof of \cref{lem:blind_re_eat}]
To show the min-entropy lower bound, we will make use of \cref{cor:alpha_chosen}.
For this, we first check that the maps $\cN_i$ satisfy the required conditions.
Since $C_i$ is a deterministic function of the (classical) variables $X_i, Y_i, A_i,$ and $B_i$, it is clear that \cref{eqn:measurement_condition} is satisfied.
For the non-signalling condition, we define the map $\cR_i \in \cptp(E_{i-1}, E_i)$ as follows: $\cR_i$ samples $T_i, X_i$ and $Y_i$ as Alice does in \cref{step:input} of \cref{prot:re}. $\cR$ then performs Eve's actions in the protocol (which only act on $Y_i$ and $E'_{i-1}$, which is part of $E_{i-1}$). 
It is clear that the distribution on $X_i$ and $Y_i$ produced by $\cR_i$ is the same as for $\cN_i$.
By the assumption that $D$ and Eve cannot communicate, the marginal of the output of $\cN_i$ on Eve's side must be independent of the device's system $R_{i-1}$.
Hence, $\setft{Tr}_{A_i R_i C_i} \circ \cN_i = \cR_i \circ \setft{Tr}_{R_{i-1}}$.

To construct a min-tradeoff function, we note that we can split $\cN_i = \gamma \cN_i^{\setft{test}} + (1-\gamma) \cN_i^{\setft{data}}$, with $\cN_i^{\setft{test}}$ always picking $T_i = 1$ and $\cN_i^{\setft{data}}$ always picking $T_i = 0$.
Then, we get from \cref{lem:infrequent_sampling} and the condition $\Max{g} = g(\delta_1)$ that the affine function $f$ defined by 
\begin{align*}
f(\delta_0) = g(\delta_1) + \frac{1}{\gamma}(g(\delta_0) - g(\delta_1)) \,, \qquad f(\delta_1) = f(\delta_\bot) = g(\delta_1)
\end{align*}
is an affine min-tradeoff function for $\{\cN_i\}$.

Viewing the event $\Omega$ as a subset of the range $\bits^n$ of the random variable $C^n$ and comparing with the abort condition in \cref{prot:re}, we see that $c^n \in \Omega$ implies $\freq{c^n}(0) \leq (1 - \wexp + \delta) \gamma$.
Therefore, for $c^n \in \Omega$ and denoting $p = \freq{c^n}$, 
\begin{align*}
f(\freq{c^n}) = p(0) f(\delta_0) + (1 - p(0)) f(\delta_1) = \frac{p(0)}{\gamma} g(\delta_0) + \left( 1- \frac{p(0)}{\gamma} \right) g(\delta_1) \geq h \,,
\end{align*}
where the last inequality holds because $g$ is affine and the distribution $p'(0) = p(0)/\gamma, p'(1) = 1-p(0)/\gamma$ satisfies $p'(0) \leq 1 - \wexp + \delta$.
The proposition now follows directly from \cref{cor:alpha_chosen} and the scaling of $c_1$ and $c_0$ is easily obtained from the expressions in \cref{cor:alpha_chosen}.

To show that an honest strategy succeeds in the protocol with high probability, we define a random variable $F_i$ by $F_i = 1$ if $C_i = 0$, and $F_i = 0$ otherwise.
If $D$ and Eve execute the quantum strategy that wins the game $G$ with probability $\wexp$ in each round, then $\E[F_i] = (1 - \wexp) \gamma$.
Using the abort condition in the protocol, we then find 
\begin{align*}
\pr{\textnormal{abort}} 
&= \pr{\sum_{i=1}^n F_i > (1 - \wexp + \delta)\cdot \gamma n} \\
&= \pr{\sum_{i=1}^n F_i > \left( 1 + \frac{\delta}{1-\wexp} \right) \cdot \E \Big[ \sum_{i=1}^n F_i \Big]} \\
&\leq e^{-\frac{\delta^2}{1 - \wexp + \delta} \gamma n} \,,
\end{align*}
where in the last line we used a Chernoff bound.
\end{proof} 

To make use of \cref{lem:blind_re_eat}, we need to construct a function $g(p)$ that satisfies the condition in \cref{eqn:g_cond}.
For this, we will use the equivalent condition \cref{eqn:g_cond_mod}. A general way of obtaining such a bound automatically is using the recent numerical method~\cite{brown2021computing}.\footnote{The main result of~\cite{miller2017randomness} (Theorem 14) does not appear to be sufficient for this. The reason is that the statement made in~\cite{miller2017randomness} essentially concerns the randomness produced on average over the question distribution $q$ of the game $G$. However, choosing a question at random consumes randomness, so to achieve exponential randomness expansion, in \cref{prot:re} we fix the inputs $x^*, y^*$ used for generation rounds. To the best of our knowledge, the results of~\cite{miller2017randomness} do not give a bound on the randomness produced in the non-local game for any \emph{fixed} inputs $x^*, y^*$. If one could prove an analogous statement to~\cite[Theorem 14]{miller2017randomness} that also certifies randomness on fixed inputs for a large class of games, our~\cref{lem:blind_re_eat} would then imply exponential blind randomness expansion for any such game. Alternatively, one can also assume that public (non-blind) randomness is a free resource and use this to choose the inputs for the non-local game. Then, no special inputs $x^*, y^*$ are needed in \cref{prot:re} to ``save randomness'' and the result of~\cite{miller2017randomness} combined with our generalised EAT implies that such a conversion from public to blind randomness is possible for any complete-support game.}
Specifically, using the assumption that Alice's lab is isolated, the maps $\cN_i$ describing a single round of the protocol take the form described in~\cref{fig:one-round-circuit}. 
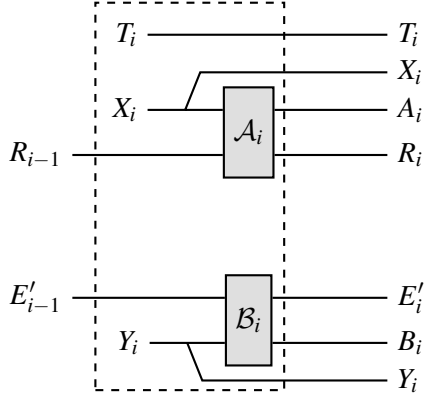
\begin{figure}[htb!]
    \centering
    \begin{tikzpicture}[thick, scale=1]
        \draw
            (0,0) node[porte, minimum height=1.2cm] (a) {$\mathcal{A}_i$}
            (a) ++(0, -2.5) node[porte, minimum height=1.2cm] (b) {$\mathcal{B}_i$}
            ;
        \draw  
            (a.west) ++(0, .3) coordinate (a-in1)
            (a.west) ++(0, -.3) coordinate (a-in2)
            (a.east) ++(0, .3) coordinate (a-out1)
            (a.east) ++(0, -.3) coordinate (a-out2)
            (b.west) ++(0, .3) coordinate (b-in1)
            (b.west) ++(0, -.3) coordinate (b-in2)
            (b.east) ++(0, .3) coordinate (b-out1)
            (b.east) ++(0, -.3) coordinate (b-out2)
            ;
        \draw
            (a.east) ++(1.5, 0) coordinate (rightwall)
            (a-in2) ++(-2, 0) coordinate (leftwall)
            ;
        \draw
            (a-in2 -| leftwall) node[left] {$R_{i-1}$} to (a-in2)
            (b-in1 -| leftwall) node[left] {$E'_{i-1}$} to (b-in1)
            (a-out2) to (a-out2 -| rightwall) node[right] {$R_{i}$}
            (b-out1) to (b-out1 -| rightwall) node[right] {$E'_i$}
            (a-in1) ++(-1,0) node[left] {$X_i$} to (a-in1)
            (b-in2) ++(-1,0) node[left] {$Y_i$} to (b-in2)
            (a-out1) to (a-out1 -| rightwall) node[right] {$A_i$}
            (b-out2) to (b-out2 -| rightwall) node[right] {$B_i$}
            (a-in1) ++(-.5, 0) to ++(.2, .5) coordinate (atmp) to (atmp -| rightwall) node[right] {$X_i$}
            (b-in2) ++(-.5, 0) to ++(.2, -.5) coordinate (btmp) to (btmp -| rightwall) node[right] {$Y_i$}
            (a-in1) ++(-1, 1) node[left] (t) {$T_i$}  to (t -| rightwall) node[right] {$T_i$}
            ;
        \node[draw, dashed, fit=(a) (b) (t) (btmp)] {};
    \end{tikzpicture}
    \vspace{-0.5em}
    \caption{Circuit diagram of $\cN: R_{i-1} E'_{i-1} \rightarrow A_iR_iT_iX_iY_iB_iE'_i$. For every round of the protocol, a circuit of this form is applied, where $\mathcal{A}$ and $\mathcal{B}$ are the (arbitrary) channels applied by Alice's device and Eve, respectively. As in the protocol, $T_i$ is a bit equal to 1 with probability $\gamma$, and $X_i$ and $Y_i$ are generated according to $q$ whenever $T_i=1$, and are fixed to $x^*,y^*$ otherwise. We did not include the register $C_i$ in the figure as it is a deterministic function of $T_iX_iY_iA_iB_i$.}
    \label{fig:one-round-circuit}
\end{figure}

\begin{figure}[!htb]
	\centering
	\vspace{0.5em}
	\begin{tikzpicture}
		
		\begin{axis}[
			width=2.5in,
			height=1.5in,
			scale only axis,
			xmin=0.75,
			xmax=0.85355,
			ymin=0,
			ymax=1.0,
			grid=major,
			xlabel={CHSH winning probability $\omega$},
			ylabel={Bits},
			xtick={0.75, 0.77, 0.79, 0.81, 0.83,0.85},
			ytick={0.0,0.2,0.4,0.6,0.8,1},
			axis background/.style={fill=white}
			]
			\addplot[color=blue, line width=1.5pt,smooth] table[col sep=comma] {blind_chsh.dat};
		\end{axis}
	\end{tikzpicture}
	\vspace{-0.7em}
	\caption{Lower bound on the conditional entropy $H(A_i|B_iX_iY_iT_iE_i')_{\rho_{|T_i=0}}$ for any state generated as in~\cref{fig:one-round-circuit} and such that on test rounds the obtained winning probability for the CHSH game is $\omega$. This lower bound was obtained by using the method from~\cite{brown2021computing}.
	For each input $y \in \cY$, the channel $\cB_{y}$ is modelled as $\cB_y(\omega) = \sum_{b} \Pi_y^{(b)} \omega \Pi_y^{(b)}$, where $\{\Pi_y^{(b)}\}_{b \in \cB}$ are orthogonal projectors summing to the identity, and similarly for the map $\cA$. It is simple to see that this is without loss of generality.}
	\label{fig:randomness_curve}
\end{figure}
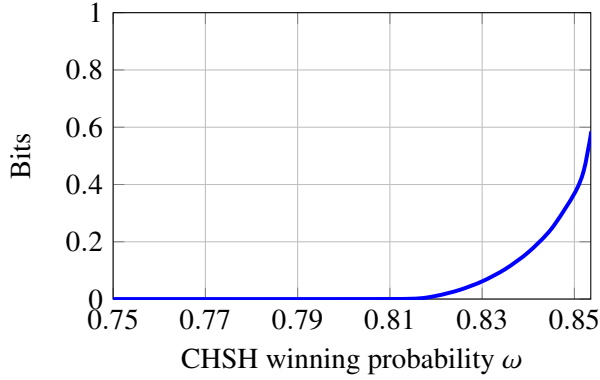 

The method of~\cite{brown2021computing} allows one to obtain lower bounds on the infimum of 
\begin{align*}
H(A_i|B_iX_iY_iT_iE'_i\tilde E_{i-1})_{\cN_i(\omega_{R_{i-1}E'_{i-1}\tilde E_{i-1}})}
\end{align*}
over all input states $\omega_{R_{i-1}E'_{i-1}\tilde E_{i-1}}$ and for any map $\cN_i$ of the form depicted in~\cref{fig:one-round-circuit}. Importantly, for any $\cN_i$ we may also restrict the infimum to states $\omega$ that are consistent with the observed statistics, i.e., $\cN^{\setft{test}}(\omega)_{C_i} = p$ for some distribution $p$ on $C_i$, using the notation of~\cref{lem:blind_re_eat}. 
Using this numerical method for the CHSH game, we obtain the values shown in~\cref{fig:randomness_curve}. 
From this, one can also construct an explicit affine min-tradeoff function $g(p)$ in an automatic way using the same method as in~\cite{brown2019framework}. 
As our focus is on illustrating the use of the generalised EAT, not the single-round bound, we do not carry out these steps in detail here.

Combining this single-round bound and \cref{lem:blind_re_eat}, one obtains that for~\cref{prot:re} instantiated with the CHSH game, $\wexp$ sufficiently close to the maximal winning probability of \smash{$\frac{1}{2} + \frac{1}{2\sqrt{2}}$}, and $\gamma = \Theta(\frac{\log n}{n})$, one can extract $\Omega(n)$ bits of uniform randomness from $A_1 \dots A_n$ while using only $\setft{polylog}(n)$ bits of randomness to run the protocol. 
In other words, \cref{prot:re} achieves exponential blind randomness expansion with the CHSH game.

\subsection{E91 quantum key distribution protocol} \label{sec:qkd}
The E91 protocol is one of the simplest entanglement-based QKD protocols~\cite{E91}.
This protocol was already treated using the original EAT in~\cite{eat}.
Here, we do not give a formal security definition and proof, only an informal comparison of how the original EAT and our generalised EAT can be applied to this problem; the remainder of the security proof is then exactly as in~\cite{eat}.
For a detailed treatment of the application of our generalised EAT to QKD, see~\cite{eat_for_qkd}.
To facilitate the comparison with~\cite{eat}, in this section we label systems the same as in~\cite{eat} even though this differs from the system labels used earlier in this paper.
The protocol we are considering is described explicitly in~\cref{prot:qkd}.
It is the same as in~\cite{eat} except for minor modifications to simplify the notation.

\begin{figure}[ht!]
\begin{longfbox}[breakable=true, padding=1em, margin-top=1em, margin-bottom=1em]
\begin{protocol} {\bf E91 quantum key distribution protocol} \label{prot:qkd} \end{protocol}
\noindent\underline{Protocol arguments}  \vspace{-0.8ex}
\vspace{-0.6ex}
\begin{center}
\begin{tabularx}{\textwidth}{r c X}
$n \in \mathbb{N}$ & : & number of uses of qubit channel \\
   $\mu \in (0,1)$ & : & probability for measurements in diagonal basis \\
  $e \in (0,\frac{1}{2})$ & : & maximum tolerated phase error ratio  \\
  $\vartheta_{\setft{EC}} \in [0,1]$ & : & relative communication cost of error correction scheme $\setft{EC}$ \\
  $r \in [0,1]$ & : & key rate  
\end{tabularx}
\end{center}

\vspace{0.1ex}

\noindent\underline{Protocol steps}

\vspace{-0.6ex}
\begin{enumerate}[label=(\arabic*)]
    \setlength{\itemsep}{0.4ex}
    \setlength{\parskip}{0.4ex}
    \setlength{\parsep}{0.4ex} 
  \item \emph{Distribution:} For $i \in \{1, \ldots, n\}$,  Alice prepares a pair $(Q_i, \bar{Q}_i)$ of entangled qubits and sends $\bar{Q}_i$ to Bob. Alice generates a random bit $B_i$ such that $P_{B_i}(1) = \mu$ and, depending on whether $B_i=0$ or $B_i = 1$, measures $Q_i$ in either the computational or  the diagonal basis, storing the outcome as $A_i$. In the same way, Bob measures $\bar{Q}_i$ in a basis determined by a random bit $\bar{B}_i$, storing the outcome as $\bar{A}_i$. 
  
  \item \emph{Sifting and information reconciliation:} Alice and Bob announce $B_i$ and $\bar{B}_i$.
  On indices $i$ where $B_i \neq \bar{B}_i$, they set $A_i = \bar A_i = \perp$.  They invoke a reliable\protect\footnotemark~error correction scheme $\setft{EC}$, allowing Bob to compute a guess $\hat{A}^n$ for Alice's string $A^n$.  If $\setft{EC}$ does not output a guess then the protocol is aborted. 
  
  \item \emph{Parameter estimation:}  Bob counts the number of indices $i \in S$  for which $\bar{B}_i = 1$ and $\bar{A}_i \neq \hat{A}_i$. If this number is larger than $e \mu^2 n$ then the protocol is aborted.

  \item \emph{Privacy amplification:} see~\cite{eat} for details.
  \end{enumerate} 
  \vspace{-1.8ex}
 \end{longfbox}
 \end{figure}

We consider the systems $B_i, \bar B_i, A_i, \bar A_i, Q_i, \bar Q_i$ as in~\cref{prot:qkd} and additionally define the system $X_i$ storing the statistical information used in the parameter estimation step:
\begin{align*}
X_i = \begin{cases}
A_i \oplus \bar A_i & \textnormal{if $B_i = \bar B_i = 1$,}\\
\bot & \textnormal{otherwise.}
\end{cases}
\end{align*}
Denoting by $E$ the side information gathered by Eve during the distribution step, we can follow the same steps as for~\cite[Equation (57)]{eat} to show that the security of~\cref{prot:qkd} follows from a lower bound on
\begin{align}
\hmin^\eps(A^n | B^n \bar B^n E)_{\rho_{|\Omega}} \,. \label{eqn:qkd_entropy_cond}
\end{align}
Here, $\rho_{|\Omega}$ is the state at the end of the protocol conditioned on acceptance.

We first sketch how the original EAT (whose setup was described in~\cref{sec:intro}) is applied to this problem in~\cite{eat}.
One cannot bound $\hmin^\eps(A^n | B^n \bar B^n E)_{\rho_{|\Omega}}$ directly using the EAT because a condition similar to~\cref{eqn:measurement_condition} has to be satisfied.
Therefore, one modifies the systems $\bar A_i$ from~\cref{prot:qkd} by setting $\bar A_i = \bot$ if $B_i = \bar B_i = 0$ and then applies the EAT to find a lower bound on
\begin{align}
\hmin^\eps(A^n \bar A^n | B^n \bar B^n E)_{\rho_{|\Omega}} \,. \label{eqn:barA}
\end{align}
For this, a round of~\cref{prot:qkd} is viewed as a map $\cM_i: Q_i^n \bar Q_i^n \to Q_{i+1}^n \bar Q_{i+1}^n A_i \bar A_i B_i \bar B_i X_i$, which chooses $B_i \bar B_i$ as in~\cref{prot:qkd}, applies Alice and Bob's (trusted) measurements on systems $Q_i \bar Q_i$ to generate $A_i \bar A_i$, and generates $X_i$ as described before.
To apply the EAT, $R_{i-1} \deq Q_i^n \bar Q_i^n$ takes the role of the ``hidden sytem'', and $A_i \bar A_i$ and $B_i \bar B_i$ are the output and side information of the $i$-th round, respectively.
It is easy to see that with this choice of systems, the Markov condition of the EAT is satisfied, so, using a min-tradeoff function derived from an entropic uncertainty relation~\cite{berta2010uncertainty}, one can find a lower bound on~\cref{eqn:barA}.

However, adding the system $\bar A_i$ in this manner has the following disadvantage: to relate the lower bound on $\hmin^\eps(A^n \bar A^n | B^n \bar B^n E)_{\rho_{|\Omega}}$ to the desired lower bound on $\hmin^\eps(A^n | B^n \bar B^n E)_{\rho_{|\Omega}}$ one needs to use a chain rule for min-entropies, incurring a penalty term of the form $\hmax^\eps(\bar A^n | A^n B^n \bar B^n E)_{\rho_{|\Omega}}$.
This penalty term is relatively easy to bound for the case of the E91 protocol, but can cause problems in general.
 \footnotetext{An error correction scheme is reliable if, except with negligible probability, either Bob's guess of Alice's string is correct or the protocol aborts.}

We now turn our attention to proving~\cref{eqn:qkd_entropy_cond} using our generalised EAT.
For this, we first observe that 
\begin{align*}
\hmin^\eps(A^n | B^n \bar B^n E)_{\rho_{|\Omega}} \geq \hmin^\eps(A^n | B^n \bar B^n X^n E)_{\rho_{|\Omega}} \,,
\end{align*}
so it suffices to find a lower bound on the r.h.s.
This step is similar to adding the $\bar A_i$ systems in~\cref{eqn:barA} in that its purpose is to satisfy~\cref{eqn:measurement_condition}.
However, it has the advantage that here, $X^n$ can be added to the \emph{conditioning} system and therefore lowers the entropy, not raises it like going from~\cref{eqn:qkd_entropy_cond} to~\cref{eqn:barA}.
The same step is not possible in the original EAT due to the restrictive Markov condition.

Using the same system names as before, we define $E_i \deq Q_{i+1}^n \bar Q_{i+1}^n B^i \bar B^i X^i E$.\footnote{In~\cref{prot:qkd}, instead of Alice distributing the systems $Q_i \bar Q_i$ and Eve gathering side information $E$ by intercepting $\bar Q_i$, we can equivalently imagine that Eve first prepares a state \smash{$\rho^0_{Q^n \bar Q^n E}$} and distributes $Q_i \bar Q_i$ to Alice and Bob in each round.
Then, the choice of $E_i$ intuitively captures the side information available to Eve from the first $i$ rounds: Eve still possess the systems \smash{$Q_{i+1}^n \bar Q_{i+1}^n$} to be distributed in future rounds, has gathered classical information $B^i \bar B^i X^i$, and keeps the static side information $E$ from preparing the initial state.}
Then, analogously to the original EAT, we can describe a single round of~\cref{prot:qkd} by a map $\cM_i: E_{i-1} \to A_i E_{i} X_i$. (Compared to the map $\cM_i$ we described above for the original EAT, we have traced out $\bar A_i$, added a copy of $X_i$, and added identity maps on the other additional systems in $E_{i-1}$.)
Denoting by \smash{$\rho^0_{Q^n \bar Q^n E}$} the joint state of Alice and Bob's systems $Q^n \bar Q^n$ before measurement and the information $E$ that Eve gathered during the distribution step, the state at the end of the protocol is $\rho = \cM_n \circ \dots \circ \cM_1(\rho^0)$.
To apply~\cref{cor:alpha_chosen} to find a lower bound on 
\begin{align*}
\hmin^\eps(A^n | E_n)_{\cM_n \circ \dots \circ \cM_1(\rho^0)_{|\Omega}} \,,
\end{align*}
we first observe that the condition in~\cref{eqn:measurement_condition} is satisfied because the system $X^n$ is part of $E_n$, and the non-signalling condition is trivially satisfied because there is no $R_i$-system.
A min-tradeoff function can be constructed in exactly the same way as in~\cite[Claim 5.2]{eat} by noting that all systems in $E_i$ on which $\cM_i$ does not act can be viewed as part of the purifying system.

This comparison highlights the advantage of the more general model of side information in our gener\-alised EAT: for the original EAT, one has to add the systems $\bar A_i$ to the ``non-conditioning side'' of the min-entropy in order to be able to satisfy the Markov condition.
In our case, the non-signalling condition, the analogue of the Markov condition, is trivially satisfied because we need no $R_i$-system. 
This is because we can add the quantum systems $Q^n \bar Q^n$ to the side information register $E_0$ at the start and then, since we allow side information to be updated and Alice and Bob act on $Q_i \bar Q_i$ using trusted measurement devices, we can remove the systems $Q_i \bar Q_i$ one by one during the rounds of the protocol.

\bibliographystyle{alpha}
\bibliography{main}

\appendix

\section{Dual statement for smooth max-entropy} \label{app:hmax}
In the main text we have focused on deriving a lower bound on the smooth min-entropy.
Here, we show that this also implies an upper bound on the smooth max-entropy by applying a simple duality relation between min-~and max-entropy.
A similar upper bound was also derived in~\cite{eat}.
However, that bound is subject to a Markov condition and cannot be derived by a simple duality argument since the ``dual version'' of the Markov condition is unwieldy.
We show that the bound from~\cite{eat} follows as a special case of our more general bound even without any Markov conditions or other non-signalling constraints.
For simplicity, we restrict ourselves to an asymptotic statement without ``testing'', i.e.~we derive an $\hmax^\eps$-version of \cref{thm:improved_eat}.
By applying the same duality relation  to the more involved statement in \cref{thm:with_testing}, one can also obtain an $\hmax^\eps$-bound with explicit constants and testing. 

Recall that for $\rho_{AB} \in \states(AB)$ and $\eps \in [0,1]$, the $\eps$-smoothed max-entropy of $A$ conditioned on $B$ is defined as
\begin{align*}
    H_{\max}^\eps(A|B)_{\rho} =  \log \inf_{\tilde \rho_{AB} \in \cB_{\eps}(\rho_{AB})} \sup_{\sigma_{B} \in \states(B)} \norm{\tilde \rho_{AB}^{\frac{1}{2}} \sigma_B^{\frac{1}{2}} }_{1}^2 \, ,
\end{align*}
where $\norm{\cdot}_{1}$ denotes the trace norm and $\cB_{\eps}(\rho_{AB})$ is the $\eps$-ball around $\rho_{AB}$ in terms of the purified distance~\cite{tomamichel2015quantum}.
The smooth min-~and max-entropy satisfy the following duality relation~\cite[Proposition 6.2]{tomamichel2015quantum}: for a pure quantum state $\psi_{ABC}$, 
\begin{align*}
\hmin^\eps(A|B)_{\psi} = - \hmax^\eps(A|C)_{\psi} \,.
\end{align*}

For the setting of \cref{thm:improved_eat}, let $V_i: R_{i-1} E_{i-1} \to A_i R_i E_i F_i$ be the Stinespring dilation of the map $\cM_i$, and let $\ket{\rho^0}_{R_0 E_0 F_0}$ be a purification of the input state $\rho^0_{R_0 E_0}$.
Then, $V_n \cdots V_1 \ket{\rho^0}$ is a purification of \smash{$\cM_n \circ \dots \circ \cM_1(\rho^0)$}, so by the duality of the smooth min-~and max-entropy, 
\begin{align*}
\hmin^\eps(A^n | E_n )_{\cM_n \circ \dots \circ \cM_1(\rho^0)} = - \hmax^\eps(A^n | F^n R_n)_{V_n \cdots V_1 \ket{\rho^0}} \,.
\end{align*}
Furthermore, by concavity of the conditional entropy the infimum in \cref{thm:improved_eat} can be restricted to pure states $\ket{\omega}_{R_{i-1} E_{i-1} \tilde E_{i-1}}$, so $V_i \ket{\omega}$ is a purification of $\cM_i(\omega)$.
Then, by the duality relation for von Neumann entropies, 
\begin{align*}
H(A_i|E_i \tilde E_{i-1} )_{\cM_i(\omega)} = - H(A_i|R_i F_i)_{V_i \ket{\omega}} \,.
\end{align*}
Therefore, we obtain the following dual statement to \cref{thm:improved_eat}:
\begin{align}
\hmax^\eps(A^n | F^n R_n)_{V_n \cdots V_1 \ket{\rho^0}} \leq \sum_{i = 1}^n \max_{\ket{\omega}} H(A_i|R_i F_i)_{V_i \ket{\omega}} + O(\sqrt{n}) \,, \label{eqn:hmax_eat}
\end{align}
where the maximisation is over pure states on $R_{i-1} E_{i-1} \tilde E_{i-1}$.
This holds for any sequence of isometries $V_i$ for which the maps $\cM_{V_i}: R_{i-1} E_{i-1} \to A_i R_i E_i$ given by $\cM_{V_i}(\rho) = \ptr{F_i}{V_i \rho V_i^\dagger }$ satisfy the non-signalling condition of \cref{thm:improved_eat}: for each $i$, there must exist a map $\cR_i \in \cptp(E_{i-1} , E_i )$ such that $\setft{Tr}_{A_i R_i} \circ \cM_{V_i} = \cR_i \circ \setft{Tr}_{R_{i-1}}$.

To gain some intuition for the above statement, consider a setting where an information source  generates systems $A_1, \dots, A_n$ and $F_1, \dots, F_n$ by applying isometries $V_i: S_{i-1} \to A_i F_i S_i$ to some pure intial state $\ket{\rho^0}_{S_0}$.
We might be interested in compressing the information in $A^n$ in such a way that given $F^n$, one can reconstruct $A^n$ except with some small failure probability $\eps$.
Then, the amount of storage needed for the compressed information is given by $\hmax^\eps(A^n|F^n)$.
To apply \cref{eqn:hmax_eat}, for $i < n$ we split the systems $S_i$ into $R_i E_i$ in such a way that the channel $\cM_{V_i}$ defined above satisfies the non-signalling condition, and set $E_n = S_n$ (so that $R_n$ is trivial).
Then \cref{eqn:hmax_eat} gives an upper bound on $\hmax^\eps(A^n|F^n)$.
Note that this bound depends on how we split the systems $S_i = R_i E_i$:
the non-signalling condition can always be trivially satisfied by choosing $R_i$ to be trivial, but \cref{eqn:hmax_eat} tells us that if we can describe the source in such a way that $E_i$ is relatively small and $R_i$ is relatively large while still satisfying the non-signalling condition, we obtain a tighter bound on the amount of required storage.

From \cref{eqn:hmax_eat} we can also recover the max-entropy version of the original EAT, but without requiring a Markov condition.
To facilitate the comparison with~\cite{eat}, we first re-state their theorem with their choice of system labels, but add a bar to every system label to avoid confusion with our notation from before.
The max-entropy statement in~\cite{eat} considers a sequence of channels $\bar \cM_i: \bar R_{i-1} \to \bar A_i \bar B_i \bar R_i$ and asserts that under a Markov condition, for any initial state $\rho_{\bar R_0 \bar E}$ with a purifying system $\bar E \equiv \bar R_0$:
\begin{align}
\hmax^\eps(\bar A^n | \bar B^n \bar E)_{\bar \cM_n \circ \cdots \circ \bar \cM_1(\rho_{\bar R_0 \bar E})} \leq \sum_{i=1}^n \max_{\omega \in \states(\bar R_{i-1} \bar R)} H(\bar A_i | \bar B_i \bar R)_{\bar \cM_i(\omega)} + O(\sqrt{n}) \,, \label{eqn:hmax_eat_old}
\end{align}
where $\bar R \equiv \bar R_{i-1}$.
We want to recover this statement from \cref{eqn:hmax_eat} \emph{without any Markov condition}.
For this, we consider the Stinepring dilations $\bar V_i: \bar R_{i-1} \to \bar R_i \bar A_i \bar B_i \bar F_i$ of $\bar \cM_i$.
We make the following choice of systems: 
\begin{align*}
R_i = \bar B^i \bar E \,, \quad A_i = \bar A_i \,, \quad E_i = \bar R_i \bar F_i \,,
\end{align*}
and choose $F_i$ to be trivial.
By tensoring with the identity, we can then extend $\bar V_i$ to an isometry $V_i: R_{i-1} E_{i-1} \to A_i R_i E_i$.
Then, the maps $\cM_{V_i}$ satisfy the non-signalling condition since $V_i$ acts as identity on $R_{i-1}$.
Therefore, remembering that $R_n = \bar B^n \bar E$ and $F^n$ is trivial, we see that \cref{eqn:hmax_eat} implies \cref{eqn:hmax_eat_old}.
Note that our derivation did not require any conditions on the channels $\bar \cM_i$ we started with, i.e.~we have shown \cref{eqn:hmax_eat_old} holds for any sequence of channels $\bar \cM_i$, not just channels satisfying a Markov or non-signalling condition.

\section{Uhlmann property for the R\'enyi divergence} \label{sec:uhlmann_alpha}
We establish that for the max-divergence (where $\alpha \to \infty$), Uhlmann's theorem holds.
\begin{proposition}
\label{prop:dmax-stable}
Let $\sigma_{A} \in \states(A)$ and $\rho_{AE} \in \states(A E)$. Then we have
\begin{align}
\label{eq:dmax-stable}
D_{\max}(\rho_{A} \| \sigma_{A}) = \inf_{\hat{\sigma}_{AR} \, : \, \hat{\sigma}_{A} = \sigma_{A}} D_{\max}(\rho_{AR} \| \hat{\sigma}_{AR}) \ .
\end{align}
In addition, if $\rho_{AR}, \rho_{A} \otimes \id_{R}$ and $\sigma_{A} \otimes \id_{R}$ all commute, then for any $\alpha \in [\frac{1}{2}, \infty)$, we have 
\begin{align}
\label{eq:classical-stable}
D_{\alpha}(\rho_{A} \| \sigma_{A}) = \inf_{\hat{\sigma}_{AR} \, : \, \hat{\sigma}_{A} = \sigma_{A}} D_{\alpha}(\rho_{AR} \| \hat{\sigma}_{AR}) \, .
\end{align}
\end{proposition}
\begin{proof}
We start with~\cref{eq:dmax-stable}. The inequality $\leq$ is a direct consequence of the data-processing inequality for $D_{\max}$. 
For the inequality $\geq$, we use semidefinite programming duality, see e.g.,~\cite{watrous2018theory}. Observe that we can write $2^{D_{\max}( \rho_{A} \| \sigma_{A})}$ as the following semidefinite program
\begin{align*}
\min_{\tau_{A} \in \pos(A), \lambda \in \R}  \{ \tr{\tau_{A}} : \rho_{A} \leq \tau_{A}, \tau_{A} = \lambda \sigma_{A} \} \, .
\end{align*}
Using semidefinite programming duality, this is also equal to 
\begin{align}
\label{eq:dual_dmaxA}
\max_{X_{A} \in \pos(A), Y_{A} \in \Herm(A) } \{ \tr{X_{A} \rho_{A}} : \id_{A} + Y_{A} = X_{A} ,  \tr{Y_{A} \sigma_{A}} = 0  \} \, .
\end{align}
We can also write a semidefinite program for $\inf_{\hat{\sigma}_{AR} \, : \, \hat{\sigma}_{A} = \sigma_A} 2^{ D_{\max}(\rho_{AR} \| \hat{\sigma}_{AR}) }$. We introduce the variable $\theta_{AR} = \lambda \hat{\sigma}_{AR}$ and get
\begin{align*}
\min_{\tau \in \pos(A \otimes R), \lambda \in \R} \{ \tr{\tau_{AR}} : \rho_{AR} \leq \tau_{AR},  \theta_{A} = \lambda \sigma_{A}  \} \, .
\end{align*}
Again, by semidefinite programming duality, we get that it is equal to
\begin{align}
\label{eq:dual_dmaxAR}
\max_{X_{AR} \in \pos(A \otimes R), Y_{A} \in \Herm(A)} \{ \tr{X_{AR} \rho_{AR}} :  (\id_{A} + Y_{A}) \otimes \id_{R} = X_{AR}, \tr{Y_{A} \sigma_{A}} = 0 \}\, .
\end{align}
Eliminating the variable $X_{AR}$, we can write this last program as 
\begin{align*}
\max_{Y_{A} \in \Herm(A)} \{ \tr{(\id_{A} + Y_{A}) \rho_{A}} :  \id_{A} + Y_{A} \in \pos(A) , \tr{Y_{A} \sigma_{A}} = 0 \}\, ,
\end{align*}
which is the same as~\cref{eq:dual_dmaxA}. This proves~\cref{eq:dmax-stable}.
\cref{eq:classical-stable} follows immediately by choosing $\hat{\sigma}_{AR} = \sigma_{A}\rho_{A}^{-1}\rho_{AR}$ and using the commutation conditions.
\end{proof}

However, for $\alpha \geq 1$ and arbitrary $\sigma_{A} \in \states(A)$, $\rho_{AE} \in \states(A E)$, the Uhlmann property given by~\cref{eq:classical-stable} does not hold. A concrete example is $\rho_{AR} = \proj{\psi}_{AR}$ with
\begin{align*}
\ket{\psi}_{AR} = \sqrt{\frac{1}{4}} \ket{00}_{AR} + \sqrt{\frac{3}{4}} \ket{11}_{AR}     
\end{align*}
and $\sigma_{A} = \frac{1}{3} \proj{+} + \frac{2}{3} \proj{-}$. In this case, $D_2(\rho_A \| \sigma_A) < 0.476$ whereas 
\begin{align*}
    \inf_{\hat{\sigma}_{AR} \, : \, \hat{\sigma}_{A} = \sigma_{A}} D_2(\rho_{AR} \| \hat{\sigma}_{AR} ) \geq \inf_{\hat{\sigma}_{AR} \, : \, \hat{\sigma}_{A} = \sigma_{A}} D(\rho_{AR} \| \hat{\sigma}_{AR}) > 0.48 \, .
\end{align*} 
This computation was performed by numerically solving the semidefinite programs via CVXQUAD~\cite{fawzi2018efficient}. Putting everything together shows that~\cref{eq:classical-stable} does not hold for $\alpha \in \{1,2\}$: 
\begin{align*}
 D(\rho_{A} \| \sigma_{A}) \leq D_2(\rho_A \| \sigma_A) 
 <  \inf_{\hat{\sigma}_{AR} \, : \, \hat{\sigma}_{A} = \sigma_{A}} D(\rho_{AR} \| \hat{\sigma}_{AR}) \leq \inf_{\hat{\sigma}_{AR} \, : \, \hat{\sigma}_{A} = \sigma_{A}} D_{2}(\rho_{AR} \| \hat{\sigma}_{AR}) \, .  
\end{align*}

\end{document}